\documentclass[journal]{IEEEtran}

\usepackage{cite}
\usepackage[pdftex]{graphicx}
\usepackage{amssymb}
\usepackage[cmex10]{amsmath}
\usepackage{amsmath,amsfonts,amssymb}
\usepackage{lineno}
\usepackage{subfigure}
\usepackage{algorithmic}
\usepackage[ruled,lined,boxed]{algorithm2e}
\usepackage{xcolor}
\usepackage{multirow}
\usepackage{diagbox}
\usepackage{threeparttable}
\usepackage{setspace}
\usepackage[justification=centering]{caption}

\newtheorem{lemma}{Lemma}
\newtheorem{remark}{Remark}

\newtheorem{proof}{Proof}

\begin{document}

\title{Active Terminal Identification, Channel Estimation, and Signal Detection for Grant-Free NOMA-OTFS in LEO Satellite Internet-of-Things}
\author{
	Xingyu Zhou, Keke Ying, Zhen Gao, Yongpeng Wu, Zhenyu Xiao, Symeon Chatzinotas, Jinhong Yuan,~\IEEEmembership{Fellow,~IEEE}, and Bj\"{o}rn Ottersten,~\IEEEmembership{Fellow,~IEEE} 
	\thanks{This paper was presented in part at the 2021 IEEE/CIC International Conference on Communications in China (ICCC) \cite{myself}.
	The work of Z. Gao is supported in part by the Natural Science Foundation of China (NSFC) under Grant 62088101, Grant 61827901, and Grant 62071044; in part by the Beijing Institute of Technology Research Fund Program for Young Scholars under
	Grant XSQD-202121009; and in part by the Ensan Foundation under Grant 2022006.
	The work of Z. Xiao is supported by the Beijing Natural Science Foundation under grant number L212003. \textit{(Corresponding author: Zhen Gao.)}
}
	\thanks{X. Zhou, K. Ying, and Z. Gao are with the School of Information and Electronics, Beijing Institute of Technology, Beijing 100081, China, and also with the Advanced Research Institute of Multidisciplinary Science, Beijing Institute of Technology, Beijing 100081, China (E-mails: \{zhouxingyu21, ykk, gaozhen16\}@bit.edu.cn).}
	\thanks{Y. Wu is with the Department of Electronic Engineering, Shanghai Jiao Tong University, Shanghai 200240, China (E-mail: yongpeng.wu@sjtu.edu.cn).}
	\thanks{Z. Xiao is with the School of Electronic and Information Engineering, Beihang University, Beijing 100191, China (E-mail: xiaozy@buaa.edu.cn).}
	\thanks{S. Chatzinotas and B. Ottersten are with Interdisciplinary Centre for Security, Reliability and Trust, University of Luxembourg, 29 Avenue J.F. Kennedy, Luxembourg City L-1855, Luxembourg (E-mails: \{symeon.chatzinotas, bjorn.ottersten\}@uni.lu).}	
	\thanks{J. Yuan is with the School of Electrical Engineering and Telecommunications, the University of New South Wales, Sydney, NSW 2025, Australia (Email: j.yuan@unsw.edu.au).}
} 
	

\maketitle

\begin{abstract}
	
This paper investigates the massive connectivity of low Earth orbit (LEO) satellite-based Internet-of-Things (IoT) for seamless global coverage.
We propose to integrate the grant-free non-orthogonal multiple access (GF-NOMA)  paradigm  with the emerging orthogonal time frequency space (OTFS) modulation to accommodate the massive IoT access, and mitigate the long round-trip latency and severe Doppler effect of terrestrial–satellite links (TSLs).
On this basis, we put forward a two-stage successive active terminal identification (ATI) and channel estimation (CE) scheme as well as a low-complexity multi-user signal detection (SD) method.
Specifically, at the first stage, the proposed training sequence aided OTFS (TS-OTFS) data frame structure facilitates the joint ATI and coarse CE, whereby both the traffic sparsity of terrestrial IoT terminals and the sparse channel impulse response are leveraged for enhanced performance.
Moreover, based on the single Doppler shift property for each TSL and sparsity of delay-Doppler domain channel, we develop a parametric approach to further refine the CE performance. 
Finally, a least square based parallel time domain SD method is developed to detect the OTFS signals with relatively low complexity.
Simulation results demonstrate the superiority of the proposed methods over the state-of-the-art solutions in terms of ATI, CE, and SD performance confronted with the long round-trip latency and severe Doppler effect.
    
\end{abstract}

\begin{IEEEkeywords}
Internet of Things (IoT), low Earth orbit (LEO) satellite, orthogonal time frequency space (OTFS), grant-free non-orthogonal multiple access (GF-NOMA).
\end{IEEEkeywords}

\IEEEpeerreviewmaketitle

\section{Introduction}\label{S1}

\IEEEPARstart{W}{ith} the advent of the 5G era, Internet of Things (IoT) based on terrestrial  cellular networks has developed rapidly and has been widely used in various aspects of human life \cite{5g-IoT}.
In the coming beyond 5G and even 6G, IoT is expected to revolutionize the way we live and work, by means of a wealth of new services based on the seamless interactions of massive heterogeneous terminals \cite{6g-IoT}. 
However, in many application scenarios, IoT terminals are widely distributed. 
Particularly, a considerable percentage of IoT terminals may be located in remote areas, which indicates that these IoT applications can not be well supported by conventional terrestrial  cellular infrastructures.
In recent years, low Earth orbit (LEO) satellite communication systems have attracted considerable research interest, and dense LEO constellations are expected to complement and extend existing terrestrial communication networks, reaching seamless global coverage.  
In fact, the commercial exploration of LEO constellations dates back to the late $20^{\rm th}$ century, such as Iridium, Globalstar, and Teledesic \cite{Liu.LEO}. Unfortunately, most of these early attempts ended with failure in the context of the underdeveloped vertical applications. 
Nowadays, a booming demand and new space technologies reignite the LEO market, where numerous enterprises, such as Starlink, OneWeb, and Telesat \cite{Liu.LEO}, envisage massive deployment. 
As an indispensable component of the 6G space-air-ground-sea integrated networks, LEO constellations are envisioned to  provide a promising solution to enable wide area IoT services {\cite{IoRT,LEO-IoT,embb}}.

Nevertheless, distinct from the terrestrial communication environment, satellite communications usually suffer from harsh channel conditions such as long round-trip delay, severe Doppler effects, poor link budget, and etc \cite{NB-IoT, Resource Allocation}.
Besides, in sharp contrast to the conventional downlink-dominated human-type communication systems, IoT is mainly driven by the uplink massive machine-type communications (mMTC) with the characteristics of sporadic behavior	
\cite{spoadic traffic}.
Meanwhile, the demands of advanced IoT-enabled applications have shifted from low-rate short packet transmission to more rigorous low-latency, broadband, and reliable information interaction \cite{6g-IoT}.
Consequently, the design of efficient random access (RA) paradigm for massive IoT terminals based on LEO satellite constitutes a challenging problem.

\subsection{Related Work}

\renewcommand{\arraystretch}{1.15}
\begin{table*}[]	
	\centering
	\caption{A brief comparison of the related literature with our work}
	\label{tab:my-table}
	\resizebox{1.8\columnwidth}{!}{%
		\begin{tabular}{|c|c|c|c|ccc|c|}
			\hline
			\multirow{2}{*}{\textbf{Reference}} & \multirow{2}{*}{\textbf{Channel model}}                                                         & \multirow{2}{*}{\textbf{Bandwidth}} & \multirow{2}{*}{\textbf{\begin{tabular}[c]{@{}c@{}}Transmit signal \\ waveform\end{tabular}}} & \multicolumn{3}{c|}{\textbf{\begin{tabular}[c]{@{}c@{}}Signal processing \\ at the receiver\end{tabular}}} & \multirow{2}{*}{\textbf{Algorithm}}                                                                               \\ \cline{5-7}
			&                                                                                                 &                                     &                                                                                               & \multicolumn{1}{c|}{\textbf{ATI}}         & \multicolumn{1}{c|}{\textbf{CE}}         & \textbf{SD}         &                                                                                                                   \\ \hline
			\cite{wang1}                        & \begin{tabular}[c]{@{}c@{}}Frequency selective\\ Rayleigh fading\end{tabular}                   & Broadband                           & OFDM                                                                                          & \multicolumn{1}{c|}{\checkmark}           & \multicolumn{1}{c|}{}                    & \checkmark          & \begin{tabular}[c]{@{}c@{}}Structured iterative \\ support detection\end{tabular}                                 \\ \hline
			\cite{Du2}                          & Rayleigh fading                                                                                 & Narrowband                          & Single-carrier                                                                                & \multicolumn{1}{c|}{\checkmark}           & \multicolumn{1}{c|}{}                    & \checkmark          & \begin{tabular}[c]{@{}c@{}}Block sparse modified \\ subspace pursuit (SP)\end{tabular}                            \\ \hline
			\cite{wang2}                        & \begin{tabular}[c]{@{}c@{}}Frequency selective\\ Rayleigh fading\end{tabular}                   & Broadband                           & OFDM                                                                                          & \multicolumn{1}{c|}{\checkmark}           & \multicolumn{1}{c|}{}                    & \checkmark          & \begin{tabular}[c]{@{}c@{}}Modified orthogonal \\ matching pursuit (OMP)\end{tabular}                             \\ \hline
			\cite{Du1}                          & Rayleigh fading                                                                                 & Narrowband                          & Single-carrier                                                                                & \multicolumn{1}{c|}{\checkmark}           & \multicolumn{1}{c|}{}                    & \checkmark          & \begin{tabular}[c]{@{}c@{}}Prior-information \\ aided adaptive SP\end{tabular}                                    \\ \hline
			\cite{JADD-MAP}                     & Rayleigh fading                                                                                 & Narrowband                          & Single-carrier                                                                                & \multicolumn{1}{c|}{\checkmark}           & \multicolumn{1}{c|}{}                    & \checkmark          & \begin{tabular}[c]{@{}c@{}}Maximum a posteriori \\ probability (MAP)\end{tabular}                                 \\ \hline
			\cite{JADD-AMP}                     & \begin{tabular}[c]{@{}c@{}}Frequency selective\\ Rayleigh fading\end{tabular}                   & Broadband                           & OFDM                                                                                          & \multicolumn{1}{c|}{\checkmark}           & \multicolumn{1}{c|}{}                    & \checkmark          & \begin{tabular}[c]{@{}c@{}}Approximate message passing \\ (AMP) and expectation \\ maximization (EM)\end{tabular} \\ \hline
			\cite{Yikunmei}                     & \begin{tabular}[c]{@{}c@{}}Frequency selective\\ Rayleigh fading\\ (Pre-equalized)\end{tabular} & Broadband                           & OFDM                                                                                          & \multicolumn{1}{c|}{\checkmark}           & \multicolumn{1}{c|}{}                    & \checkmark          & \begin{tabular}[c]{@{}c@{}}Orthogonal AMP with \\ multiple measurement \\ vectors (MMV)\end{tabular}              \\ \hline
			\cite{JAUDCE}                       & Frequency fading                                                                                & Broadband                           & OFDM                                                                                          & \multicolumn{1}{c|}{\checkmark}           & \multicolumn{1}{c|}{\checkmark}          &                     & \begin{tabular}[c]{@{}c@{}}Iterative identified \\ user cancellation\end{tabular}                                 \\ \hline
			\cite{C-RAN}                        & Rayleigh fading                                                                                 & Narrowband                          & Single-carrier                                                                                & \multicolumn{1}{c|}{\checkmark}           & \multicolumn{1}{c|}{\checkmark}          &                     & Modified Bayesian CS                                                                                              \\ \hline
			\cite{AMP}                          & Rayleigh fading                                                                                 & Narrowband                          & Single-carrier                                                                                & \multicolumn{1}{c|}{\checkmark}           & \multicolumn{1}{c|}{\checkmark}          &                     & AMP                                                                                                               \\ \hline
			\cite{KML}                          & \begin{tabular}[c]{@{}c@{}}Frequency selective \\ fading\end{tabular}                           & Broadband                           & OFDM                                                                                          & \multicolumn{1}{c|}{\checkmark}           & \multicolumn{1}{c|}{\checkmark}          &                     & \begin{tabular}[c]{@{}c@{}}Generalized MMV \\ (GMMV)-AMP-EM\end{tabular}                                          \\ \hline
			\cite{LEO.IoT}                      & Land mobile satellite                                                                           & Narrowband                          & Single-carrier                                                                                & \multicolumn{1}{c|}{\checkmark}           & \multicolumn{1}{c|}{\checkmark}          &                     & Bernoulli–Rician MP-EM                                                                                            \\ \hline
			\cite{SWQ.OTFS}                     & Double-dispersive                                                                               & Broadband                           & OTFS                                                                                          & \multicolumn{1}{c|}{}                     & \multicolumn{1}{c|}{\checkmark}          & \checkmark          & \begin{tabular}[c]{@{}c@{}}Three-dimensional \\ simultaneous-OMP\end{tabular}                                     \\ \hline
			\cite{ZS-OTFS-CE}                   & Double-dispersive                                                                               & Broadband                           & OTFS                                                                                          & \multicolumn{1}{c|}{}                     & \multicolumn{1}{c|}{\checkmark}          &                     & EM-variational Bayesian (VB)                                                                                      \\ \hline
			Our work                            & \begin{tabular}[c]{@{}c@{}}TSL channel\\ (Double-dispersive)\end{tabular}                       & Broadband                           & OTFS                                                                                          & \multicolumn{1}{c|}{\checkmark}           & \multicolumn{1}{c|}{\checkmark}          & \checkmark          & \begin{tabular}[c]{@{}c@{}}Two-stage ATI \& CE and\\ LS-based parallel SD\end{tabular}                            \\ \hline
		\end{tabular}%
	}
\end{table*}


The traditional grant-based RA protocols adopted by terrestrial cellular networks usually suffer from the complicated control signaling exchanges and scheduling for requesting uplink access resources \cite{granted-based,IoRT}. 
In the case of the extremely long terrestrial–satellite link (TSL) and the resulting large round-trip signal propagation delay, this type of solution will further aggravate the access latency. 
To this end, the ALOHA protocols arise as a better option and 
are widely used in existing satellite communications for RA \cite{RA-Sat}.
The original ALOHA protocol allows the terminals to transmit their data packets without any coordination.
To improve the RA throughput, more advanced ALOHA techniques are developed, such as contention resolution diversity ALOHA (CRDSA) \cite{CRDSA} and enhanced spread spectrum ALOHA (E-SSA), and etc. 
Despite the aforementioned efforts, the current ALOHA-based RA protocols mainly depend on orthogonal multiple access (OMA) technique and may suffer from the network congestion when the number of terrestrial IoT terminals becomes massive \cite{RA-Sat}.

Recently, grant-free non-orthogonal multiple access (GF-NOMA) schemes have been emerging. 
These schemes allow IoT terminals to directly transmit their non-orthogonal preambles followed by data packets over the uplink and avoid complicated access requests for resource scheduling \cite{Grant-free}.
By exploiting the intrinsic sporadic traffic, the receiver of the base station (BS) can 
separate the non-orthogonal preambles transmitted by different terminals and thus identify the active terminal set (ATS)  with compressive sensing (CS) techniques.
Benefitting from the non-orthogonal resource allocation, the GF-NOMA schemes can improve the system throughput with limited radio resources. 
To date, the state-of-the-art CS-based GF-NOMA study mainly focuses on two typical problems:
1) joint active terminal identification (ATI) and signal detection (SD); 
2) joint ATI and channel estimation (CE). 
 
The former category is developed by assuming the perfect channel state information (CSI) known at the BS \cite{wang1,wang2,JADD-AMP,JADD-MAP,Du1,Du2} or the perfect pre-equalization 
at the terminals (e.g., based on the beacons periodically broadcast by the BS \cite{Yikunmei}),
where CSI is usually regarded to be \textit{quasi-static}.
In particular, \cite{wang1} and \cite{Du2} proposed a
structured iterative support detection algorithm and a block sparsity based subspace pursuit (SP) algorithm, respectively, to jointly perform ATI and SD in one signal frame (consists of multiple continuous time slots), where the terminals' activity is assumed to remain unchanged.
\cite{wang2} and \cite{Du1} further relaxed the assumption, i.e., the ATS may vary in several continuous time slots,
and developed a modified OMP algorithm and \textit{a priori} information aided adaptive SP algorithm, respectively, to perform dynamic ATI and SD, where the estimated ATS is exploited as \textit{a priori} knowledge for the identification in the following time slots.
Moreover, to fully exploit the \textit{a priori} information of the transmit constellation symbols for enhanced accuracy, some Bayesian inference-based detection algorithms were proposed in \cite{JADD-AMP,JADD-MAP,Yikunmei}. In \cite{JADD-MAP}, based on the maximum \textit{a posteriori} probability (MAP) criterion, the proposed algorithm calculated \textit{a posteriori} activity probability and soft symbol information to identify the active terminals and detect their payload data, respectively.
To overcome the challenge that the perfect \textit{a priori} information could be unavailable in practical systems, an approximate message passing (AMP)-based scheme was proposed in \cite{JADD-AMP}, where the hyper-parameters 
of terminals' activity and noise variance can be adaptively learned through the expectation-maximization (EM) algorithm.
The above literature is mainly based on the assumption that the CSI is perfectly known at the BS and requires the elements of the adopted spreading sequences to be independent and identically distributed (i.i.d), which can be unrealistic in practice. 
Therefore, \cite{Yikunmei} developed an orthogonal AMP (OAMP)-based ATI and SD algorithm for orthogonal frequency division multiplexing (OFDM) systems, where the CSI can be pre-equalized at the terminals according to the beacon signals broadcast by the BS, and the spreading sequences are selected from the partial discrete Fourier transformation (DFT) matrix. 
 
Another category can be applied to \textit{time-varying} channels, where perfect CSI at the BS or perfect pre-compensation at terminals is unrealistic \cite{JAUDCE,AMP,C-RAN,LEO.IoT,KML}. 
An iterative joint ATI and CE scheme was proposed in \cite{JAUDCE}, where the sparsity of delay-domain channel impulse response (CIR) was exploited and an identified user cancellation approach was proposed for enhanced performance.
By exploiting not only the sparse traffic behavior of IoT terminals, but also the innate heterogeneous path loss effects and the joint sparsity structures in multi-antenna systems,  the authors in \cite{C-RAN} developed a modified Bayesian CS algorithm.
With the full knowledge of the \textit{a priori} distribution of the
channels and the noise variance, the authors in \cite{AMP} developed an AMP-based scheme for massive access in massive multiple-input multiple-output (MIMO) systems. 
For more challenging massive MIMO-OFDM systems, the authors in \cite{KML} proposed a generalized multiple measurement vector (GMMV)-AMP algorithm, where  
the structured sparsity of spatial-frequency domain and angular-frequency domain channels were leveraged with EM algorithm incorporated. 
Moreover, a Bernoulli–Rician message passing with expectation–maximization
(BR-MP-EM) algorithm was proposed for the LEO satellite-based narrowband massive access using single-carrier in \cite{LEO.IoT}.
However, these aforementioned works \cite{JAUDCE,AMP,C-RAN,KML} 
usually assume the channels to be \textit{slowly time-varying}, which can not be directly applied to the highly dynamic TSLs due to the high-mobility of LEO satellites. 
For clarity, the comparison of the aforementioned related works on GF-NOMA is summarized in Table \ref{tab:my-table}.

\subsection{Motivation}

An emerging two-dimensional modulation scheme, orthogonal time frequency space (OTFS), has been widely considered as a promising alternative to the dominant OFDM.
Particularly, OTFS is expected to support reliable communications under high-mobility scenarios in the next-generation mobile communications \cite{OTFS,OTFS Pilot,OFDM-based OTFS, OTFS preamble, OTFS-MA, OTFS-NOMA, OTFS-NOMA syst, OTFS-PDMA, OTFS-fractional Doppler, OTFS-mag, SWQ.OTFS,ZS-OTFS-CE,yjh}.
OTFS multiplexes information symbols on a lattice in the delay-Doppler (DD)
domain and utilizes a compact DD channel model, where the channel in the DD domain is considered to exhibit more stable, separable, and sparse features than that in the TF domain \cite{OTFS-mag}.
Consequently, OTFS can achieve more robust signal processing with additional diversity gain in the presence of Doppler effect \cite{OTFS, OTFS-mag}.
In fact, \cite{OTFS-MA, OTFS preamble, OTFS-PDMA} have integrated the OTFS waveform with OMA 
based on the grant-based access protocols and investigated some new resource allocation schemes. Besides, \cite{OTFS-NOMA, OTFS-NOMA syst} further amalgamated the OTFS modulation scheme with the NOMA technique.   
However, \cite{OTFS-MA,OTFS preamble,OTFS-PDMA,OTFS-NOMA,OTFS-NOMA syst} adopt the \textit{grant-based} RA schemes, which may not cater to the stringent requirements of access latency and massive connectivity for LEO satellite-based IoT.
Moreover, the study of CE for OTFS is only limited to terminals employing OMA scheme \cite{SWQ.OTFS, OTFS Pilot, ZS-OTFS-CE, OTFS-fractional Doppler}.

\subsection{Contributions}

{In this paper, we propose a GF-NOMA paradigm that incorporates OTFS modulation to provide a blend of mMTC and enhanced mobile broadband (eMBB) services for LEO satellite-based IoT, and investigate the challenging ATI, CE, and SD problems.}
The main contributions of this paper are summarized as follows.
\begin{itemize}
	\item  \textbf{GF-NOMA-OTFS paradigm}:
	We propose to apply the GF-NOMA scheme employing OTFS
	waveform (GF-NOMA-OTFS) to LEO satellite-based massive IoT access. By allowing the
	uncoordinated IoT terminals to transmit the data packets directly, reusing the limited delay-Doppler resources, and exploiting the particular stability, sparsity, and separability of TSLs represented in the DD domain, the GF-NOMA-OTFS paradigm can harvest the benefit of high RA throughput and Doppler-robustness in this context.

	\color{black}
	\item \textbf{Training sequences (TSs) aided OTFS modulation/demodulation architecture}:
	Existing CE solutions for OTFS systems embed the pilot and guard symbols in the DD domain \cite{OTFS Pilot,SWQ.OTFS, ZS-OTFS-CE, OTFS-fractional Doppler}.		
	However, in the case of highly dynamic TSLs with extremely severe Doppler shifts, 
	the compactness of the DD domain channel could be destroyed, which would give rise to a dramatical increase of guard symbols. 
	Moreover, the low-resolution of Doppler lattices could lead to the severe Doppler spreading even each TSL's Doppler shift is a single value, which would further deteriorate the performance and effectiveness of signal processing in the DD domain.
	To circumvent these challenges, we utilize the time domain TSs to replace the conventional DD domain pilot and guard symbols for performing ATI and CE, and further propose a TSs aided OTFS (TS-OTFS) modulation/demodulation architecture.	
	
	\item \textbf{Successive ATI, CE, and SD method}: 	
	Furthermore, we put forward a two-stage successive ATI and CE scheme as well as a following low-complexity multi-user SD for the GF-NOMA-OTFS paradigm.
	Specifically, for the ATI and CE, at the first stage, the proposed time domain TSs facilitate the joint ATI and coarse CE, whereby both the traffic sparsity of terrestrial IoT terminals and the structural sparse CIR are leveraged.
	On this basis, a parametric approach is developed to further refine the CE performance,
	whereby the single Doppler shift property for each TSL and sparsity of DD domain channel are exploited. 
	Finally, a least square (LS)-based parallel time domain SD is developed for detecting the OTFS signals with relatively low complexity.
	Simulations and performance evaluation are conducted to varify the effectiveness of the proposed successive ATI, CE, and SD method.
\end{itemize}

\subsection{Organization}

The remainder of this paper is organized as follows. In Section \ref{S2}, 
we introduce the TSL model of the LEO satellite-based IoT.
The GF-NOMA-OTFS paradigm and TS-OTFS modulation/demodulation architecture are proposed in 
Section \ref{S3}.
In Section \ref{S4}, the proposed successive ATI and CE scheme for the GF-NOMA-OTFS paradigm is presented. 
Then, in Section \ref{S5}, we further propose a multi-user signal detector based on the previous results of ATI and CE. 
The effectiveness of our proposed scheme is demonstrated by simulation results in Section \ref{S6}. 
Finally, conclusions are drawn in Section \ref{S7}. 
{The important variables of the system model adopted in the paper are listed in Table I for ease of reference.}

\subsection{Notations}

Throughout this paper, scalar variables are denoted
by normal-face letters, while boldface lower and upper-case letters
denote column vectors and matrices, respectively. 
The transpose, Hermitian transpose, inversion, and pseudo-inversion for matrix
are denoted by $(\cdot)^{\rm T}$, $(\cdot)^{\rm H}$, $(\cdot)^{-1}$, and $(\cdot)^{\dagger}$, respectively.
Besides, $|\cdot|$, $\Vert \cdot \Vert_2$, and $\Vert \cdot \Vert_{\rm F}$ represent modulus, $\ell_2$-norm, and Frobenius-norm, respectively.
$\mathbf{X}_{[m,n]}$ is the $(m,n)$-th element of matrix $\mathbf{X}$; 
$\mathbf{X}_{[m,:]}$ and $\mathbf{X}_{[:,n]}$ are the $m$-th row vector and the $n$-th column vector of matrix $\mathbf{X}$, respectively. 
$\mathbf{X}_{[:,\mathcal{I}]}$ and $\mathbf{X}_{[:,\mathcal{I}]}$ denote the submatrix consisting of the columns and rows of $\mathbf{X}$ indexed by the ordered set $\mathcal{I}$, respectively.
$\mathbf{x}_{[n]}$ denotes the $n$-th element of $\mathbf{x}$. 
Furthermore, $|\mathcal{A}|_c$ is the cardinality of the set $\mathcal{A}$, and $\mathrm{supp}(\cdot)$ is the support set of a vector or a matrix. 
$\mathcal{A} - \mathcal{B}$ denotes the difference set of $\mathcal{A}$ with respect to $\mathcal{B}$, and $\mathcal{A} \cap \mathcal{B}$ denotes the intersection of $\mathcal{A}$ and $\mathcal{B}$.
The operators $\odot$ and $\otimes$ represent the Hadamard product and Kronecker product, respectively. 
The operator $\mathrm{vec}(\mathbf{X})$ stacks the
columns of $\mathbf{X}$ on top of each another, and $\mathrm{mat}(\mathbf{x};m,n)$ converts the vector $\mathbf{x}$ of size $mn$ into the matrix of size $m \times n$ by successively selecting every $m$ elements of $\mathbf{x}$ as its columns.
$<\mathbf{x}_1,\mathbf{x}_2>$ represents the inner product of $\mathbf{x}_1$ and $\mathbf{x}_2$. Finally, $\mathbf{I}_n$ is the identity matrix of size $n \times n$, $\mathbf{0}_{n \times m}$ is the $n \times m$ zero matrix, $\emptyset$ denotes the empty set, and $\delta(x)$ is the Dirac function.

\begin{table}[t]
	\centering
	\captionsetup{font={color = {black}}, justification = raggedright,labelsep=period}
	\caption{Variable list}	
	\label{tab:my-table5}
	\resizebox{0.95\columnwidth}{!}{%
		\begin{tabular}{cl}
			\hline
			\textbf{Notation} & \multicolumn{1}{c}{\textbf{Defination}} \\ \hline
			$P^s$,$P^t$ & Number of antennas for satellite and terminal \\ \hline
			$K$,$K_a$ & Number of potential and active terminals \\ \hline
			$\alpha_k$,$\mathcal{A}$ & Activity indicator and active terminal set \\ \hline
			$\gamma_k$,$Q_k$ & Rician factor, number of NLoS paths \\ \hline
			$g_k^{\rm LoS}$,$g_k^q$ & Complex path gain of LoS and NLoS paths \\ \hline
			$\tau_k^{\rm LoS}$,$\tau_k^q$ & RToA and propagation delay of NLoS paths \\ \hline
			$\nu^{\rm LoS}_k$,$\nu^q$ & Doppler shift \\ \hline
			$\theta^{\rm zen}_k$,$\theta^{\rm azi}_k$,$\theta^{{\rm zen}^{\prime}}$,$\theta^{{\rm azi}^{\prime}}$ & Zenith and azimuth angles of Rx and Tx \\ \hline
			$\mathbf{v}_R$,$\mathbf{v}_T$ & Steering vector of Rx and Tx \\ \hline
			$g_k^{\rm ABF-LoS}$,$g_k^{{\rm ABF}-q}$ & Analog beamforming gain \\ \hline
			$\mathbf{H}_k^{\rm DD}(\tau,\nu)$ & DD domain uplink CIR \\ \hline
			$\mathbf{h}_k^{\rm eff-DD}(\tau,\nu)$ & DD domain effective baseband CIR \\ \hline
			$\mathbf{h}^{\rm eff}_k(t,\tau)$ & Time-delay domain CIR \\ \hline
			$h^{\rm eff}_{k,p}[\kappa,\ell]$,$h^{\rm eff}_{k,p}[\ell,\upsilon]$ & Discrete time-delay domain and DD domain CIR \\ \hline
			$(M,N,M_t)$ & Size of TS-OTFS frame \\ \hline
			$G$,$L$ & Length of non-ISI region and CIR \\ \hline
			$\Delta f$,$B_w$,$T_s$,$T$ & \begin{tabular}[c]{@{}l@{}}Subcarrier spacing, bandwidth, sampling interval\\ frame duration\end{tabular} \\ \hline
			$\mathbf{X}_k^{\rm DD}$,$\mathbf{X}_k^{\rm TF}$,$\tilde{\mathbf{S}}_k$ & Data symbols in the DD, TF, and time domain \\ \hline
			$\tilde{\mathbf{s}}_k$,$\mathbf{c}_k$,$\mathbf{s}_k$ & \begin{tabular}[c]{@{}l@{}}Time domain data symbols, training sequences, \\ and transmit signal\end{tabular} \\ \hline
			$\mathbf{r}_p$,$\hat{\tilde{\mathbf{s}}}_k$, & Received signal, detected OTFS payload data \\ \hline
			$\hat{\tilde{\mathbf{S}}}_k$,$\hat{\mathbf{X}}^{\rm TF}$,$\hat{\mathbf{X}}^{\rm DD}$ & Detected time, TF, and DD domain data symbols \\ \hline
		\end{tabular}%
	}
\end{table}

\newcounter{mytempeqncnt}
\begin{figure*}[t]
	\normalsize
	
	\begin{align}\label{eq:CIR}
		\begin{split}
			\mathbf{H}^{\rm DD}_k(\tau, \nu)  = & \sqrt{\frac{\gamma_k}{\gamma_k+1}} 
			g_{k}^{\rm LoS}  \delta (\tau-\tau_k^{\rm LoS}) \delta(\nu-\nu_k^{\rm LoS})
			\mathbf{v}_{R}(\theta^{\rm zen}_{k,\rm LoS}, \theta^{\rm azi}_{k,\rm LoS}) \mathbf{v}_{T}^{\rm H}(\theta^{{\rm zen}^{\prime}}_{k,\rm LoS}, \theta^{{\rm azi}^{\prime}}_{k,\rm LoS})
			\\ & +  \sqrt{\frac{1}{\gamma_k+1}} \sum_{q=1}^{Q_k} g_{k}^{q}  \delta(\tau-\tau_{k}^{q})\delta(\nu-\nu_k^{q}) \mathbf{v}_R (\theta^{\rm zen}_{k,q}, \theta^{\rm azi}_{k,q}) \mathbf{v}_{T}^{\rm H}(\theta^{{\rm zen}^{\prime}}_{k,q}, \theta^{{\rm azi}^{\prime}}_{k,q}),
		\end{split}	
	\end{align}	 
	\setcounter{equation}{3}
	\begin{align}\label{eq:eff-DD}
		\begin{split}
		\mathbf{h}^{\rm eff-DD}_k(\tau, \nu) & =  \sqrt{\frac{\gamma_k}{\gamma_k+1}} g_k^{\rm LoS} g_k^{\rm ABF-LoS}   \delta (\tau-\tau_k^{\rm LoS}) \delta(\nu-\nu_k^{\rm LoS})
		\mathbf{v}_{R}(\theta^{\rm zen}_{k}, \theta^{\rm azi}_{k}) \\
		& +  \sqrt{\frac{1}{\gamma_k+1}} \sum_{q=1}^{Q_k} g_{k}^{q} g_k^{{\rm ABF}-q}  \delta(\tau-\tau_{k}^{q})\delta(\nu-\nu_k^{q}) \mathbf{v}_R (\theta^{\rm zen}_{k}, \theta^{\rm azi}_{k}),
		\end{split}
	\end{align}
	\hrulefill
\end{figure*}
\setcounter{equation}{1}

\section{Terrestrial-Satellite Link Model}\label{S2}

As illustrated in Fig.~\ref{fig1:system model}, {we consider that 
the LEO constellations consisting of a large number of LEO satellites can 
provide ubiquitous connections and eMBB services for massive IoT terminals\footnote{As we mainly focus on RA for one satellite in the physical layer in this paper for obtaining important RA design insights, the problems of interference and coordination between multiple satellites will not be involved \cite{Liu.LEO}.} \cite{IoRT,LEO-IoT,embb}.}
Each LEO satellite is equipped with a uniform planar array (UPA) composed of  $P^s = P_x^s \times P_y^s$ antennas, where $P_x^s$ and $P_y^s$ are the number of antennas on the x-axis and y-axis, respectively.
Meanwhile, a $P^t = P_x^t \times P_y^t$ phased array with analog beamforming is assumed to be employed at each IoT terminal.
Due to the sporadic traffic behavior in typical IoT \cite{spoadic traffic}, within a given time interval, the number of active IoT terminals $K_a$ can be much smaller than the number of all potential IoT terminals $K$, i.e., $K_a \ll K$. 
The active IoT terminals transmit RA signals and the inactive remain silent.
To reflect the activity status of all potential IoT terminals, we define an activity indicator $\alpha_k$, which is equal to 1 when the $k$-th IoT terminal is active and 0 otherwise. Meanwhile, the ATS is defined as $\mathcal{A} = \{k|\alpha_k = 1, 1 \le k \le K\}$ and its cardinality is $K_a = |\mathcal{A}|_c$.

Since analog beamforming at the IoT terminals can be approximately implemented with the predictable trajectory of LEO satellites in theory, the TSLs connecting the LEO satellites and terrestrial IoT terminals would experience few propagation scatterers and the line-of-sight (LoS) links could rarely be blocked by obstacles \cite{NB-IoT}.
It is reasonable to assume there coexist the LoS and few non-LoS (NLoS) links when employing the X-band and above.  
Therefore, the DD domain uplink channel between the LEO satellite and the $k$-th served  IoT terminal can be expressed as Eq. (\ref{eq:CIR}) \cite{You.LEO, OTFS-mag},
where the first term corresponds to the LoS path and the NLoS paths contribute to the other $Q_k$ terms.
In (\ref{eq:CIR}), $\nu_{k}^{\rm LoS}$ and $\nu_{k}^q$ respectively denote the Doppler shift of the LoS and the $q$-th NLoS path, $\tau_{k}^{\rm LoS}$ and $\tau_{k}^q$ respectively denote the remanent relative time of arrive (RToA) and delay of the $q$-th NLoS path, $\gamma_k$ and $g_{k}^{\rm LoS} (g_{k}^q)$ are respectively the Rician factor and complex path gain, 
$\mathbf{v}_{R}(\theta^{\rm zen}_{k,\rm LoS}, \theta^{\rm azi}_{k,\rm LoS}) \,
\left( \mathbf{v}_R (\theta^{\rm zen}_{k,q}, \theta^{\rm azi}_{k,q}) \right) \in \mathbb{C}^{P^s \times 1}$ and  $\mathbf{v}_{T}(\theta^{{\rm zen}^\prime}_{k,\rm LoS}, \theta^{{\rm azi}^\prime}_{k,\rm LoS}) \,
\left( \mathbf{v}_T (\theta^{{\rm zen}^\prime}_{k,q}, \theta^{{\rm azi}^\prime}_{k,q}) \right) \in \mathbb{C}^{P^t \times 1}$ denote the UPA's steering vector for the LEO satellite and IoT terminal, respectively. 
The further explanations of these parameters are detailed as follows.

\begin{figure}[t]	
	\centering
	\includegraphics[width=\columnwidth, keepaspectratio]{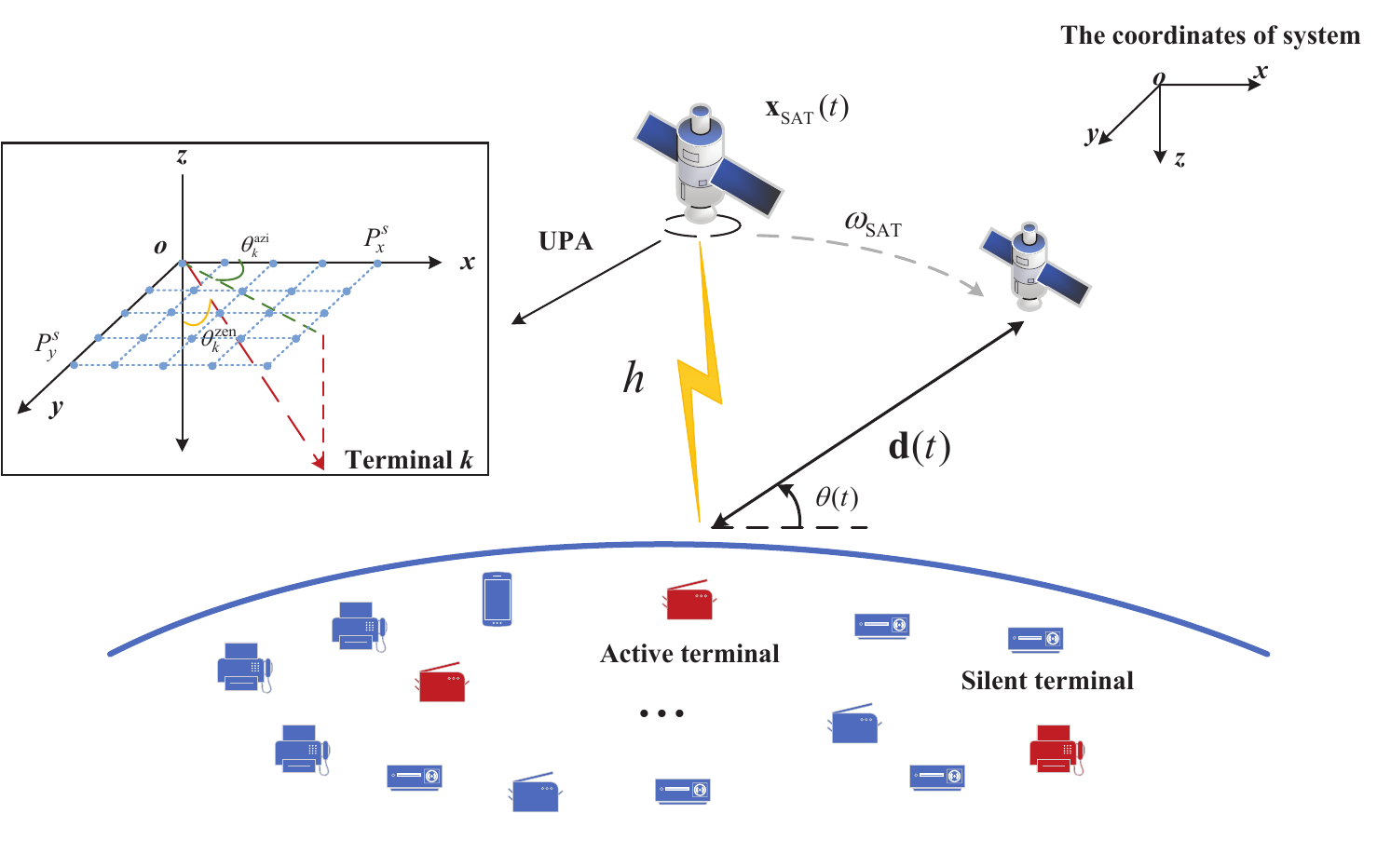}
	\captionsetup{font={color = {black}}, justification = raggedright,labelsep=period}
	\caption{Illustration of the LEO satellite-based IoT based on the proposed GF-NOMA-OTFS scheme.}
	\label{fig1:system model}	
\end{figure}

\begin{itemize}	 	
	\item \textbf{Array steering vector}:
	Since the TSL's distance is far larger than the distances between the terminal and its surrounding scatterers, the angles of arrival (AoAs), i.e., the zenith angle $\theta_{k,{\rm LoS}}^{\rm zen} \, (\theta_{k,q}^{\rm zen})$ and the azimuth angle $\theta_{k,{\rm LoS}}^{\rm azi} \, (\theta_{k,q}^{\rm azi})$, for the $k$-th terminal can be assumed to be almost identical \cite{You.LEO}, i.e., $ \theta_{k,{\rm LoS}}^{\rm zen} \approx \theta_{k,{q}}^{\rm zen}  = \theta_{k}^{\rm zen}$ 
	and $ \theta_{k,{\rm LoS}}^{\rm azi} \approx \theta_{k,{q}}^{\rm azi}   = \theta_{k}^{\rm azi}$. 
	Therefore, the UPA's steering vector for the LEO satellite can be simplified as 
	\begin{align}\label{eq:array steering vector}
		\begin{split}
			\mathbf{v}_R(\theta_{k}^{\rm zen},\theta_{k}^{\rm azi}) = \frac{1}{\sqrt{P^s}} \left[e^{-j2\pi \frac{d}{\lambda}  \sin(\theta^{\rm zen}_k) \cos(\theta^{\rm azi}_k)  \mathbf{p}_1}\right]   
			\\ \otimes \left[e^{-j2\pi \frac{d}{\lambda} \sin(\theta^{\rm zen}_k) \sin(\theta^{\rm azi}_k)  \mathbf{p}_2} \right] 
		\end{split},
	\end{align}	
	where $\mathbf{p}_1 = \left[0, 1, \dots, P_x^s-1\right]^{\rm T}$, $\mathbf{p}_2 = \left[0, 1, \dots, P_y^s-1\right]^{\rm T}$, $\lambda$ is the wavelength of carrier frequency and $d$ is the antenna spacing. Without loss of generality, the elements of the UPA are assumed to be separated by one-half wavelength in both the x-axis and y-axis.  Beisdes, $\mathbf{v}_{T}(\theta^{{\rm zen}^\prime}_{k,\rm LoS}, \theta^{{\rm azi}^\prime}_{k,\rm LoS}) \,\left( \mathbf{v}_T (\theta^{{\rm zen}^\prime}_{k,q}, \theta^{{\rm azi}^\prime}_{k,q}) \right)$ shares a similar expression with $\mathbf{v}_R(\theta_{k}^{\rm zen},\theta_{k}^{\rm azi})$.
	
	\item \textbf{Doppler shift}: 
	The Doppler shift $\nu_{k}^{\rm LoS}(\nu_{k}^q)$ includes two independent components: $\nu_{k}^{\rm LoS-S}$ $(\nu_{k}^{q-{\rm S}})$ and $\nu_{k}^{\rm LoS-T} \, (\nu_{k}^{q-{\rm T}})$ caused by the mobility of LEO satellite and terrestrial IoT terminals, respectively.
	Since LEO satellite moves much faster than IoT terminals, the motion of LEO satellite mainly determines 	$\nu_k^{\rm LoS} \, (\nu_k^q)$, i.e.,  $ \nu_k^{{\rm LoS}-{\rm S}} \gg \nu_k^{{\rm LoS-T}} \, (\nu_k^{q-{\rm S}} \gg \nu_k^{q-{\rm T}})$.
    Besides, combined with the fact that the AoAs of LoS and NLoS links related to the $k$-th terminal are alomst identical, it is reasonable to assume that the Doppler shift of the TSL is single-valued, i.e., $\nu_k^{\rm LoS} \approx \nu_k^q \approx \nu_{k}^{\rm LoS-Sat} \approx \nu_{k}^{q-{\rm Sat}}$ \cite{LMS,NB-IoT}.
     	
	\item \textbf{Remanent RToA and multipath components' (MPCs') delay}:
	Since IoT terminals’ locations are geographically distributed, the ToA of signals received from different terminals may undergo severe time offsets.
	Although \cite{syn} proposed a repetition code spreading scheme without  synchronization and scheduling, its cost is significant throughput reduction and intra-system interference. In contrast, we consider the major part of time offsets can be compensated by timing advance \cite{TA}, while the remanent RToA is denoted as $\tau_{k}^{\rm LoS}$. Meanwhile, in the case of MPCs, the relative delay of the $q$-th NLoS path can be denoted as $\tau_{k}^q$.  	
\end{itemize}

\begin{figure*}[t]	
	\centering
	\includegraphics[width=1.85\columnwidth, keepaspectratio]{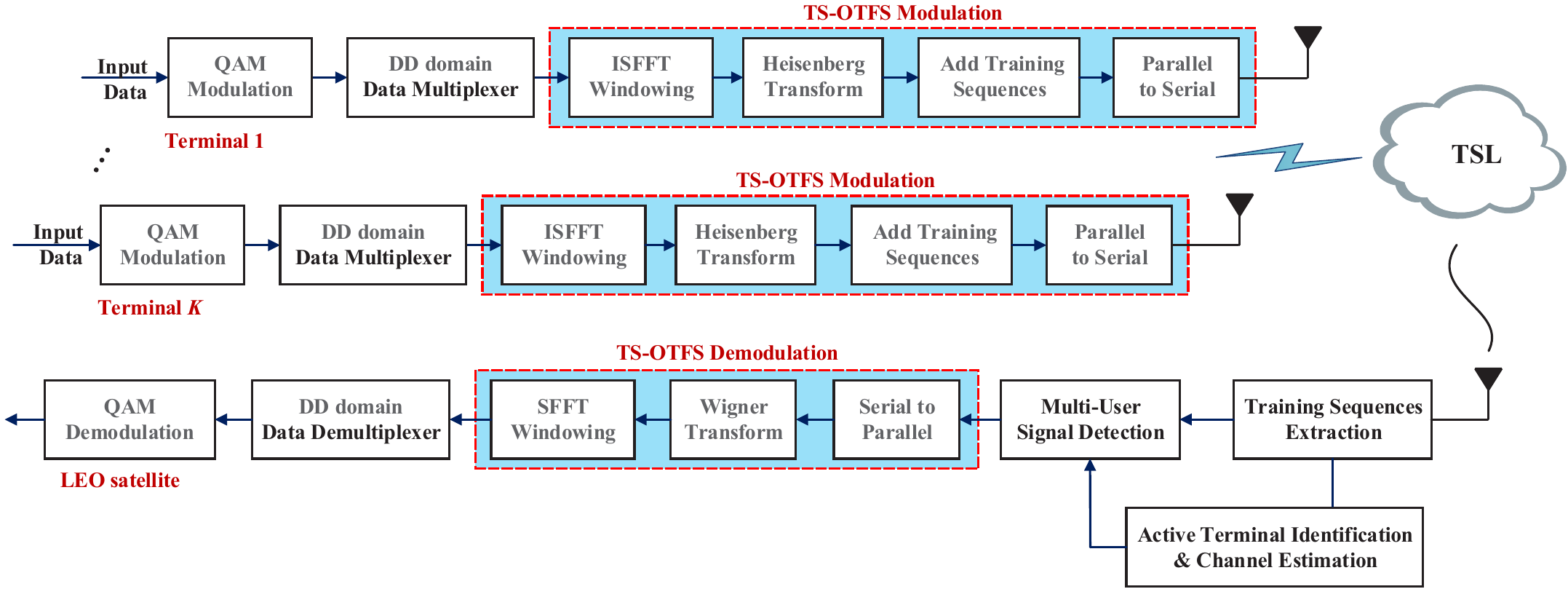}
	\captionsetup{font={color = {black}}, justification = raggedright,labelsep=period}
	\caption{The transceiver structure of the proposed TS-OTFS scheme for the GF-NOMA-OTFS paradigm.}
	\label{fig2:system flowchart}
\end{figure*}

After taking the analog beamforming at the terminals into consideration, the effective baseband channel can be written as 
\begin{align}
	\mathbf{h}^{\rm eff-DD}_k(\tau, \nu) = & \mathbf{H}^{\rm DD}_k(\tau, \nu) \mathbf{v}_{T}(\theta^{{\rm zen}^{\prime\prime}}_{k}, \theta^{{\rm azi}^{\prime\prime}}_{k}), 
\end{align}
where $\mathbf{v}_{T}(\theta^{{\rm zen}^{\prime\prime}}_{k}, \theta^{{\rm azi}^{\prime\prime}}_{k}) \in \mathbb{C}^{P^t \times 1}$ is the analog beamforming vector. From Eq. (\ref{eq:CIR}), $\mathbf{h}^{\rm eff-DD}_k(\tau, \nu)$ can be further represented by Eq. (\ref{eq:eff-DD}), where $g_k^{\rm ABF-LoS} = \mathbf{v}_{T}^{\rm H}(\theta^{{\rm zen}^{\prime}}_{k,\rm LoS}, \theta^{{\rm azi}^{\prime}}_{k,\rm LoS}) \mathbf{v}_{T}(\theta^{{\rm zen}^{\prime\prime}}_{k}, \theta^{{\rm azi}^{\prime\prime}}_{k}) $ and $g_k^{{\rm ABF}-q} = \mathbf{v}_{T}^{\rm H}(\theta^{{\rm zen}^{\prime}}_{k,q}, \theta^{{\rm azi}^{\prime}}_{k,q}) \mathbf{v}_{T}(\theta^{{\rm zen}^{\prime\prime}}_{k}, \theta^{{\rm azi}^{\prime\prime}}_{k})$ are the analog beamforming gain.
\setcounter{equation}{4}
Meanwhile, note that Eq. (\ref{eq:eff-DD}) can be transformed into time-varying CIR through 
\begin{align} \label{eq:CIR2}
	\mathbf{h}_k^{\rm eff}(t, \tau) = \int \mathbf{h}_k^{\rm eff-DD}(\tau, \nu) e^{j 2\pi \nu (t-\tau)} d \nu.
\end{align}

\section{Proposed TS-OTFS Transmission Scheme}\label{S3}

In this section, we introduce the GF-NOMA-OTFS paradigm and the transceiver structure of the proposed TS-OTFS scheme, which is illustrated in Fig.~\ref{fig2:system flowchart}.

\subsection{Modulation of the Proposed TS-OTFS at Transmitter}\label{S3.1}

For the active IoT terminals, the input information bits are first mapped to quadrature amplitude modulation (QAM) symbols and then rearranged in the DD domain plane as $\mathbf{X}^{\rm DD}_k \in \mathbb{C}^{M \times N}, \forall k$. 
Here, $N$ and  $M$ are the dimensions of the latticed resource units in the Doppler domain and delay domain, respectively.
On this basis, the DD domain $\mathbf{X}^{\rm DD}_k$ is parallel-to-serial converted to the transmit signal vector $\mathbf{s}_k$ in the time domain via a cascade of TS-OTFS transformations, which are constituted by a pre-processing module and time-frequency (TF) modulator. 
 
Specifically, the pre-processing module is consistent with that of the traditional OFDM-based OTFS architecture \cite{OFDM-based OTFS}, i.e., the DD domain data $\mathbf{X}^{\rm DD}_k$ is transformed into the TF domain data matrix $\mathbf{X}^{\rm TF}_k \in \mathbb{C}^{M \times N}$ by applying the \textit{inverse symplectic finite Fourier transform (ISFFT)}, which can be written as 
\begin{align}\label{eq:ISFFT}
	\mathbf{X}^{\rm TF}_k = \mathbf{F}_{M}\mathbf{X}^{\rm DD}_k\mathbf{F}_{N}^{\rm H}, \forall k, 	
\end{align}	
where both $\mathbf{F}_{M} \in \mathbb{C}^{M \times M}$ and $\mathbf{F}_{N} \in \mathbb{C}^{N \times N}$ are the DFT matrices. 
Based on the acquired TF domain data matrix $\mathbf{X}^{\rm TF}_k$, the subsequent TF modulator transforms $\mathbf{X}^{\rm TF}_k$ into the transmit signal vector $\mathbf{s}_k$. 
In particular, firstly, \textit{Heisenberg transform} \cite{OTFS} is applied to each column of $\mathbf{X}^{\rm TF}_k$ to produce the time domain data matrix $\tilde{\mathbf{S}}_k \in \mathbb{C}^{M \times N}$ as  
\begin{align}\label{eq:IFFT} 
	\tilde{\mathbf{S}}_k =  \mathbf{W}^{\rm tx} \odot (\mathbf{F}_{M}^{\rm H}\mathbf{X}^{\rm TF}_k), \forall k. 	 
\end{align}	
For simplicity and without loss of generality, a rectangular window, namely $\mathbf{W}^{\rm tx}$ with all elements equal to one, is adopted in this paper. In this case, the \textit{Heisenberg transform} degenerates into Fourier transform.


Furthermore, for the traditional OFDM-based OTFS architecture, a  cyclic prefix (CP) is added to the front of each time domain OTFS symbol $\tilde{\mathbf{s}}_{k}^i \in \mathbb{C}^{M \times 1}$ ($\tilde{\mathbf{s}}_{k}^i$ is the $i$-th column vector of  $\tilde{\mathbf{S}}_k$).
By contrast, for the proposed TS-OTFS scheme, $N+1$ duplicate TSs with the length of $M_t$, denoted by $\mathbf{c}_k = \left[c_{k,0} ~ c_{k,1} \dots ~ c_{k,M_{t}-1}\right]^{\rm T} \in \mathbb{C}^{M_t \times 1}$, are appended to the front and rear of
the OTFS payload data as illustrated in Fig.~\ref{fig3:data frame}.
These time domain TSs are known by the transceiver, and they can not only be utilized to avoid inter-symbol-interference (ISI) over time dispersive channels, but also perform ATI and CE (will be detailed in Section \ref{S4}). 
Finally, the transmit TS-OTFS signal consisting of TSs and time domain OTFS payload data  $\mathbf{s}_k  = \left[ \mathbf{c}_k^{\rm T}, \tilde{\mathbf{s}}_{k}^{1^{\rm T}}  , \mathbf{c}_k^{\rm T}  , \tilde{\mathbf{s}}_{k}^{2^{\rm T}} , \dots , \mathbf{c}_k^{\rm T}  , \tilde{\mathbf{s}}_{k}^{N^{\rm T}}  , \mathbf{c}_k^{\rm T}   \right]^{\rm T} \in \mathbb{C}^{(M_tN+MN+M_t) \times 1}$ is obtained through the  parallel-to-serial conversion.    

\begin{figure*}[t]	
	\centering
	\includegraphics[width=1.9\columnwidth, keepaspectratio]{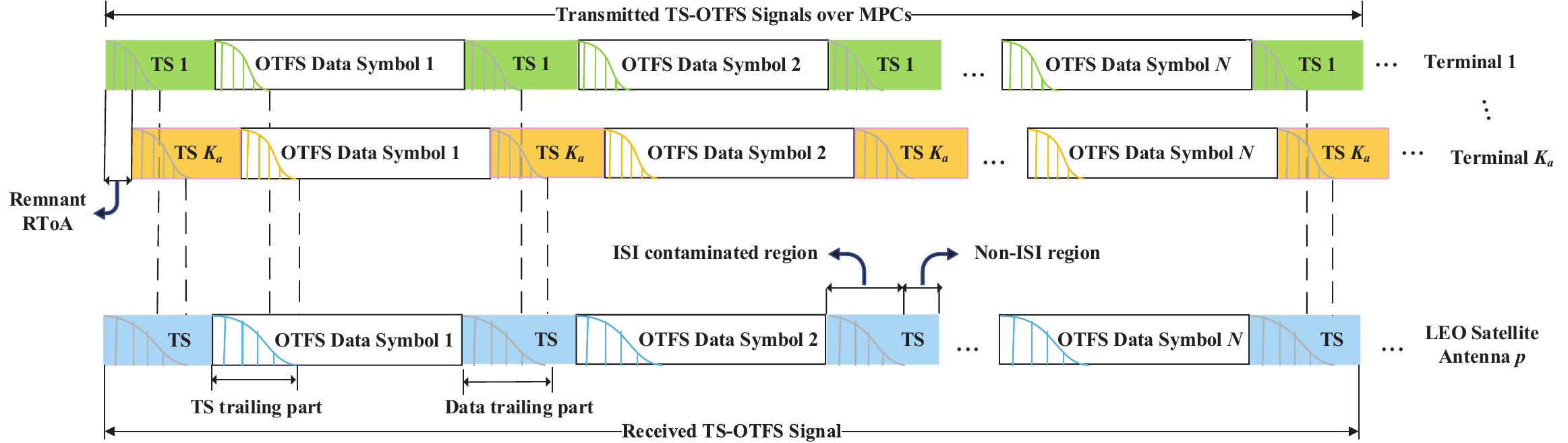}
	\captionsetup{font={color = {black}}, justification = raggedright,labelsep=period}
	\caption{Transmit and receive signal structures of the proposed TS-OTFS scheme at the transceiver.}
	\label{fig3:data frame}	
\end{figure*}

\subsection{Proposed TS-OTFS Demodulation at Receiver}\label{S3.2}

In fact, the discrete form of Eq. (\ref{eq:CIR2}) can be rewritten as
\begin{align}\label{eq:T-CIR}
	\begin{split}
		h_{k,p}^{\rm eff}[\kappa,\ell] = h_{k,p}^{\rm eff}(t,\tau)\!\mid_{t = \kappa T_s, \, \tau = \ell T_s} 	
	\end{split},
\end{align}
where $h_{k,p}^{\rm eff}(t,\tau)$ is the $p$-th element of $\mathbf{h}_{k}^{\rm eff} (t,\tau)$ and $T_s$ is the sampling interval of the system.
Moreover, the discrete form of Eq. (\ref{eq:eff-DD}) can be denoted as\footnote{Due to the large value of the sampling interval $T$ in the Doppler domain, the fractional part of normalized Doppler \textit{can not} be neglected \cite{yjh}, which means the normalized Doppler shift $\upsilon_k^{\rm LoS} = \nu_k^{\rm LoS} NT$ (also $\upsilon_k^{q}$) tends to be off-gird and comprises integer and fractional components.}
\begin{align}\label{eq:discrete CIR}
	\begin{split}
		h^{\rm eff-DD}_{k,p}[\ell, \upsilon] = h^{\rm eff-DD}_{k,p}(\tau,\nu)\!\mid_{\tau = \frac{\ell}{M \Delta f}, \, \nu = \frac{\upsilon}{NT} } 
	\end{split},
\end{align}
where $h^{\rm eff-DD}_{k,p} (\tau,\nu)$ is the $p$-th element of $\mathbf{h}_{k}^{\rm eff-DD} (\tau,\nu)$,
$\Delta f$ is the frequency spacing between adjacent sub-carriers, and $T = (M+M_t)T_s$ is the duration of one TS-OTFS symbol. 

Therefore, the $\kappa$-th element of the signal $\mathbf{r}_p \in \mathbb{C}^{(M_tN+MN+M_t) \times 1}$ received at the $p$-th antenna is the superposition of the signals received from all active terminals, which can be expressed as 
\begin{align}\label{eq:IO}
	r_p(\kappa) = \sum_{k=1}^K \sum_{l=0}^{L-1} \alpha_k  \sqrt{P_k}  h_{k,p}^{\rm eff}\left[\kappa,\ell \right] s_k\left[\kappa-\ell \right] + w_p(\kappa), \forall p, 
\end{align} 
where $P_k$ denotes the transmit power of the $k$-th terminal, $L-1$ represents the maximum of remanent RToA and MPCs' delay, and $w_p(\kappa) \sim \mathcal{CN}(0,\sigma^2_w)$ denotes the additive white Gaussian noise (AWGN) at the receiver.

The receiver of LEO satellite consists of two cascaded modules: the first one performs ATI, CE, and multi-user SD, and the others demodulate the OTFS payload data.
For the first one, the receiver of LEO satellites firstly extracts TSs from the received signals to perform ATI and CE. 
With the identified active terminal set (ATS) $\hat{\mathcal{A}}$ and their corresponding CSI, 
the proposed multi-user signal detector detects the payload data for the ATS to obtain  $\hat{\tilde{\mathbf{s}}}_k \in \mathbb{C}^{MN \times 1}, k \in \hat{\mathcal{A}}$. The above ATI, CE, and SD modules will be discussed in detail in the following Sections \ref{S4} and \ref{S5}. 

For the TS-OTFS demodulation, it is equivalent to the inverse operation of the modulation, which consists of a TF demodulator and a post-processing module, and transforms the detected time domain OTFS payload data $\hat{\tilde{\mathbf{s}}}_k$ to the original DD domain $\hat{\mathbf{X}}_{k}^{\rm DD}$.
In particular, $\hat{\tilde{\mathbf{s}}}_k$ can be rewritten as time domain 2D data matrix $\hat{\tilde{\mathbf{S}}}_k \in \mathbb{C}^{M \times N}$ through serial-to-parallel conversion, i.e.,  
\begin{align}\label{eq:SPC}  
	\hat{\tilde{\mathbf{S}}}_k =  {\rm mat}\left( \hat{\tilde{\mathbf{s}}}_k ,M,N \right ), \forall k,
\end{align}
Then, the \textit{Wigner transform} \cite{OTFS} is applied to recover the TF data $\hat{\mathbf{X}}^{\rm TF}_{k}$ as
\begin{align}\label{eq:FFT}    
	\hat{\mathbf{X}}^{\rm TF}_{{k}} =  \mathbf{W}^{\rm rx} \odot (\mathbf{F}_M \hat{\tilde{\mathbf{S}}}_k), \forall k.
\end{align}
where a rectangular window $\mathbf{W}^{\rm rx}$ is adopted similar to the transmitter in Eq. (\ref{eq:IFFT}).
In the post-processing module, \textit{symplectic finite Fourier transform (SFFT)} is applied to $\hat{\mathbf{X}}^{\rm TF,W}_{k}$ 
for restoring the TF domain OTFS data to DD domain as
\begin{align}\label{eq:SFFT}
	\hat{\mathbf{X}}_{k}^{\rm DD} = \mathbf{F}_M^{\rm H}\hat{\mathbf{X}}^{\rm TF,W}_{k}\mathbf{F}_N, \forall k.	
\end{align}

\section{Proposed Active Terminal Identification and Channel Estimation}\label{S4}

To handle the challenging ATI and CE over TSLs with severe Doppler effect, 
we propose a two-stage successive ATI and CE scheme for the GF-NOMA-OTFS paradigm.

\begin{figure*}[t]	
	\centering
	\subfigure[]{\includegraphics[width=0.8\columnwidth, keepaspectratio]{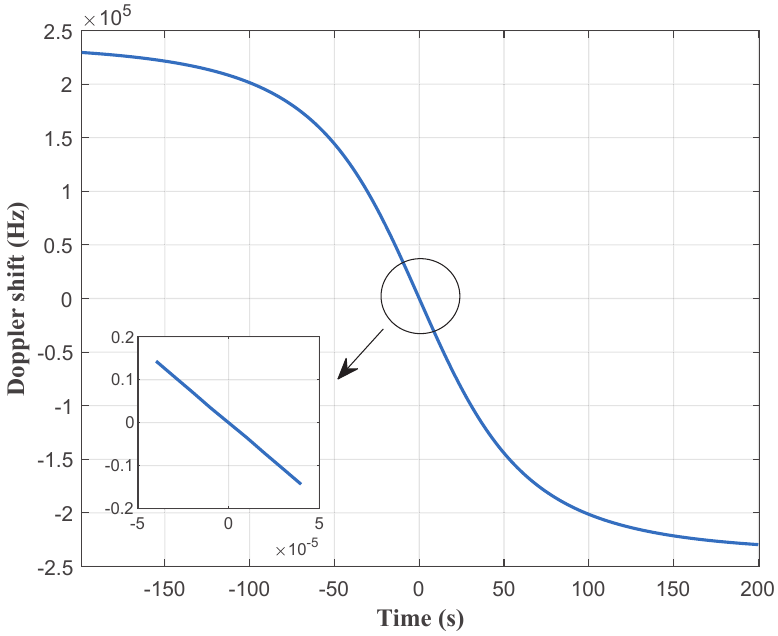}}
	\subfigure[]{\includegraphics[width=0.8\columnwidth, keepaspectratio]{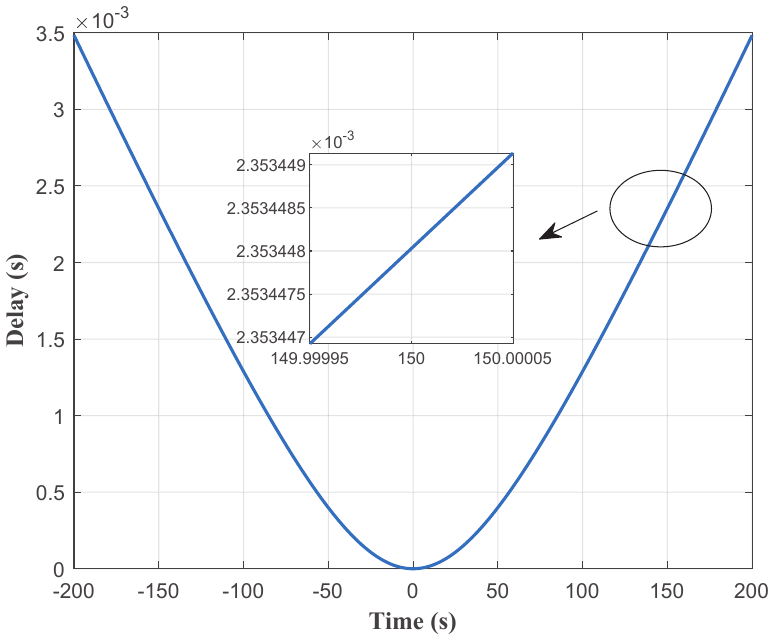}}
	\captionsetup{font={color = {black}}, justification = raggedright,labelsep=period}
	\caption{The Doppler and delay variation as a function of time: (a) Doppler shift; (b) Relative delay.} 
	\label{fig:Doppler shift}		
\end{figure*} 

\subsection{Problem Formulation of Successive ATI and CE}\label{S4.1}

The structure of time domain TS-OTFS signal vector $\mathbf{s}_k \, (1 \le k \le K)$ and its received version at the $p$-th receive antenna  $\mathbf{r}_p \, (1 \le p \le P^s)$ are illustrated in Fig. \ref{fig3:data frame}, both of which consist of the OTFS payload data and the embedded TSs part.
One significant challenge to perform ATI and CE based on the received TSs, lies in the fact that each received TS is contaminated by the previous OTFS data symbol due to the time dispersive CIR of each TSL and the remnant RToA  among different terminals' TSLs.
An effective approach is to utilize the non-ISI region as illustrated in Fig. \ref{fig3:data frame}, which is the rear part of the TSs and immune from the influence of the previous OTFS data symbol \cite{Gao.TDS-OFDM}.
Therefore, the TS's length $M_t$ is designed to be longer than {the maximum of remnant RToA and MPCs' delay $L-1$} in order to ensure the non-ISI region with sufficient length, and thus the length of non-ISI region can be denoted as $G \triangleq M_t - L + 1$.

In this way, according to (\ref{eq:IO}), the non-ISI region of the $i$-th $ (1 \le i \le N+1)$ TS $\mathbf{r}_{{\rm TS},p}^{i} \in \mathbb{C}^{G \times 1}$ can be expressed as 
\begin{align}\label{eq:vec IO} 
	 \mathbf{r}_{{\rm TS},p}^{i} =  & \sum_{k=1}^K \alpha_k  \sqrt{P_k} \left( \mathbf{\Delta}_k^{\rm LoS} \mathbf{\Psi}_{k} \mathbf{h}_{{\rm TS},k,p}^{{\rm eff}-i,{\rm LoS}} + \sum_{q=1}^{Q_k} \mathbf{\Delta}_k^{q}    \mathbf{\Psi}_{k}  \mathbf{h}_{{\rm TS},k,p}^{{\rm eff}-i,q} \right) \nonumber \\ 
	&  + \mathbf{w}_{{\rm TS},p}^{i}, \forall i, p, 
\end{align}
where $\mathbf{h}_{{\rm TS},k,p}^{{\rm eff}-i,{\rm LoS}} \in \mathbb{C}^{L
	 \times 1}$ and  $\mathbf{h}_{{\rm TS},k,p}^{{\rm eff}-i,q} \in \mathbb{C}^{L \times 1}$  denote the LoS and NLoS components of the vector form of CIR $\mathbf{h}_{{\rm TS},k,p}^{{\rm eff}-i}$ (aligned with the instant of the beginning of the $i$-th non-ISI region) as
\begin{align}
	\mathbf{h}_{{\rm TS},k,p}^{{\rm eff}-i} = \mathbf{h}_{{\rm TS},k,p}^{{\rm eff}-i,{\rm LoS}} + \sum_{q=1}^{Q_k}  \mathbf{h}_{{\rm TS},k,p}^{{\rm eff}-i,q}, \forall i,p,k,
\end{align}
$\mathbf{w}_{{\rm TS},p}^{i} \in \mathbb{C}^{G \times 1}$ is the vector form of AWGN, $\mathbf{\Psi}_k \in \mathbb{C}^{G \times L}$ is a Toeplitz matrix given by \cite{Gao.TDS-OFDM} 
\begin{align}\label{eq:Toeplitz matrix} 
	\mathbf{\Psi}_k =& \left[ \begin{array}{cccc}
		c_{k,L-1} & c_{k,L-2} & \cdots & c_{k,0} \\
		c_{k,L} & c_{k,L-1} & \cdots & c_{k,1} \\
		\vdots & \vdots & \ddots & \vdots \\
		c_{k,M_t-1} & c_{k,M_t-2} & \cdots & c_{k,M_t-L} \\
	\end{array} \right],
\end{align}	
$\mathbf{\Delta}_k^{\rm LoS} \, (\mathbf{\Delta}_k^{q}) \in \mathbb{C}^{G \times G}$ is the diagonal Doppler shift matrix associated with the LoS (the $q$-th NLoS) path as $	 \mathbf{\Delta}_k^{\rm LoS} = \mathrm{diag}\left\{   
e^{ \frac{j2\pi\upsilon_k^{\rm LoS}}{N(M+M_t)} \cdot \left[ (-\ell_k^{\rm LoS}), \dots,
0, \dots,  (G-\ell_k^{{\rm LoS}}-1)  \right]^{\rm T} }	\right\}$, and $\mathbf{\Delta}_k^{q}$ shares a similar expression.

Since both $\mathbf{\Delta}_k^{\rm LoS}$ and $\mathbf{\Delta}_k^{q}$ are unknown matrices for the receiver of LEO satellites, it would be infeasible to recover the sparse CIR vectors in Eq. (\ref{eq:vec IO}) with the unknown sensing matrices.
Fortunately, on one hand, the duration of each non-ISI region is always short enough. When the TSLs are assumed to be unchanged in this region, the approximation error could not be obvious for the support set estimation in the sparse CIR vector recovery (which will be verified through simulations in Section \ref{S6}). 
On the other hand, it will be clarified in Remark \ref{R2} that the ambiguity of recovered non-zero elements caused by this approximation can be compensated through the following CE refinement.  
In this case, both $\mathbf{\Delta}_k^{\rm LoS}$ and $\mathbf{\Delta}_k^{q}$ are approximate to the identity matrices, and (\ref{eq:vec IO}) can be rewritten as 
\begin{align}\label{eq:approx vec} 	 
		\mathbf{r}_{{\rm TS},p}^{i} & = \sum_{k=1}^K \alpha_k  \sqrt{P_k}  \mathbf{\Psi}_{k}  \underbrace{ \left( \mathbf{h}_{{\rm TS},k,p}^{{\rm eff}-i,{\rm LoS}} + \sum_{q=1}^{Q_k}  \mathbf{h}_{{\rm TS},k,p}^{{\rm eff}-i,q} \right) }_{\mathbf{h}_{{\rm TS},k,p}^{{\rm eff}-i}} + \tilde{\mathbf{w}}_{{\rm TS},p}^{i} \nonumber \\
	 	& = \mathbf{\Psi} \tilde{\mathbf{h}}_{{\rm TS},p}^{{\rm eff}-i} + \tilde{\mathbf{w}}_{{\rm TS},p}^{i},
\end{align}
where $\tilde{\mathbf{h}}^{{\rm eff}-i}_{{\rm TS},p} = \left[ \tilde{\mathbf{h}}_{{\rm TS},1,p}^{{{\rm eff}-i}^{\rm T}}, \tilde{\mathbf{h}}_{{\rm TS},2,p}^{{{\rm eff}-i}^{\rm T}}, \dots, \tilde{\mathbf{h}}_{{\rm TS},K,p}^{{{\rm eff}-i}^{\rm T}} \right]^{\rm T} \in \mathbb{C}^{KL \times 1}$, $\tilde{\mathbf{h}}_{{\rm TS},k,p}^{{\rm eff}-i} = \alpha_k \sqrt{P}_k \mathbf{h}_{{\rm TS},k,p}^{{\rm eff}-i}$,
$\mathbf{\Psi} = \left[  \mathbf{\Psi}_1, \mathbf{\Psi}_2,  \dots,    \mathbf{\Psi}_K\right] \in \mathbb{C}^{G \times KL}$, 
the approximation error and the AWGN are collectedly considered as the effective noise term  $\tilde{\mathbf{w}}_{{\rm TS},p}^{i}$.

Due to the severe path loss of TSLs, the energy of NLoS paths reflected by scatterers around the terminals could be weak and the number of non-negligible NLoS paths is limited.
Therefore, the delay domain sparsity of the CIR vector $\mathbf{h}_{{\rm TS},k,p}^{{\rm eff}-i}$ can be represented as    
\begin{align}      
	\left|\mathrm{supp}\left\{\mathbf{h}_{{\rm TS},k,p}^{{\rm eff}-i}\right\}\right|_c = Q_k+1 \ll L, \forall i,p,k.
\end{align}
Moreover, combined with the sporadic traffic behaviors, $\tilde{\mathbf{h}}_{{\rm TS},p}^{{\rm eff}-i}$ exhibits the sparsity as  
\begin{align}\label{eq:21}
	\left|\mathrm{supp}\left\{\tilde{\mathbf{h}}_{{\rm TS},p}^{{\rm eff}-i}\right\}\right|_c = \sum_{k \in \mathcal{A}}(Q_k+1) = Q \ll KL, \forall i,p. 	
\end{align} 
It indicates that (\ref{eq:approx vec}) is a typical sparse signal recovery problem with the single measurement vector (SMV) form.

\subsection{Exploiting TSL's Spatial and Temporal Correlations for Enhanced Performance}\label{s} 

To further enhance the system performance, we will explore the structured common sparsity inherent in the TSLs. 
\textit{1) Spatial correlation:}
Specially, for different receive antennas, the RToA, MPCs' delay, and Doppler shift of the signals received from the same terminal are approximately identical. 
This implies that the support sets can be treated as common for sparse CIR vectors $\{ \tilde{\mathbf{h}}_{{\rm TS},k,p}^{{\rm eff}-i} \}_{p=1}^P$, whereas the non-zero coefficients could be distinct. 
\textit{2) Temporal correlation:} 
Additionally, although the TSLs vary continuously with time due to the high mobility of LEO satellites, within the duration of one frame,
the relative positions of the IoT terminals and LEO satellite will not change dramatically. This fact implies that it can also be reasonable to assume the RToA, the propagation delay, and the Doppler shift of signals received from the
same terminal are  approximately identical for multiple adjacent  OTFS symbols within one TS-OTFS frame.  
Hence, the support sets can also be regarded as identical for sparse CIR vectors $\{ \tilde{\mathbf{h}}_{{\rm TS},k,p}^{{\rm eff}-i} \}_{i=1}^{N+1}$.

We will further illustrate the TSL's spatial and temporal correlations with the following example. 
As illustrated in Fig.~\ref{fig1:system model}, it assumes a Cartesian coordinate system such that the moving satellite and the transmitter of the terminal are on the x-z plane. The Doppler shift experienced by a stationary terminal can be computed as follows as a function of time:
\begin{align}
	f_d(t) = \frac{f_c}{c} \frac{\mathbf{d}(t)}{|\mathbf{d}(t)|} \frac{\partial \mathbf{x}_{\rm SAT} (t)}{\partial t},
\end{align}
where $f_c$ is the carrier frequency, $\mathbf{d}(t)$ is the distance vector between the satellite and the terminal, and $\mathbf{x}_{\rm SAT}(t)$ is the vector of the satellite position. These vectors can be expressed as: 
\begin{align}
	& \mathbf{d}(t) = \left[ (R_E+h)\sin(\omega_{\rm SAT} t) ,  0 , R_E - (R_E+h)\cos(\omega_{\rm SAT} t)    \right]^{\rm T}, \\
	& \mathbf{x}_{\rm SAT} (t) = \left[  (R_E+h)\sin(\omega_{\rm SAT} t) , 0 , -(R_E+h)\cos(\omega_{\rm SAT} t) \right]^{\rm T},
\end{align}
where $R_E$ is the Earth radius, $h$ is the satellite altitude, and $\omega_{\rm SAT}$ is the satellite angular velocity. 
The relative delay experienced by a stationary terminal can be computed as well:
\begin{align}
	\tau_{\rm rel} = \frac{|\mathbf{d}(t)|-|\mathbf{d}_0(t)|}{c},
\end{align}
where $\mathbf{d}_0(t)$ is the distance vector when the satellite is closest to the receiver.
After some mathematical manipulation \cite{3GPP}, the Doppler shift as a function of the elevation angle is computed in a closed-form expression as follows:
\begin{align}
	f_d(t) = \frac{f_c}{c} \omega_{\rm SAT} R_E \cos[\theta(t)],
\end{align}
where $c$ is the light velocity.

Fig.~\ref{fig:Doppler shift} illustrates the  Doppler and delay variation as a function of time, where the carrier frequency is $f_c = 10 \ \textrm{GHz}$, and the orbital altitude of LEO satellite is $h = 500 \rm km$. 
It can be observed when $t=0$, the gradient of Doppler shift function reaches its maximum value, and when $t=150$, the gradient of delay function reaches a relative maximum value.
Assuming the size of TS-OTFS frame is $(M,N,M_t)=(256,8,50)$ and the system bandwidth $B_w = 100  \ \rm MHz$. 
Then the resulting maximum Doppler and delay jitter, which are defined as the maximum variation of delay and Doppler parameters in the duration of one TS-OTFS frame $\Delta t = 25 \ \mu \rm s$,  
are $|\Delta f_d| = 0.09 \ {\rm Hz}$ and $|\Delta \tau_{\rm rel}| = 0.55 \ {\rm ns}$, respectively. In fact, when employing the subcarrier spacing $\Delta f = 480 \ \rm kHz$, the Doppler jitter is $2 \times 10^{-5} \%$ of subcarrier spacing, which is negligible for inter-carrier interference. 
Meanwhile, the delay jitter occupies only $5.5\%$ of the sampling interval $T_s = 1/B_w = 1 \times 10^{-8} \ \rm s$. 
Therefore, the spatial and temporal correlation of TSL channels in the DD domain can be  guaranteed. 

\color{black}
\begin{remark}
	The above analysis implies that the stability of TSLs in the DD domain can maintain in the one TS-OTFS frame since both the delay and Doppler jitters are negligible, and thus it determines that the DD domain signal processing of OTFS is effective in this case.
\end{remark}

With the above discussion in mind, we come to the conclusion that the CIR vectors display a common sparsity pattern across the time and spatial domain.  
On this basis, we propose to extend Eq. (\ref{eq:approx vec}) to MMV to jointly process the received signal from multiple antennas and multiple TSs. 
Specifically, by collecting the received signal $\{ \mathbf{r}_{{\rm TS},p}^{i} \}_{p=1}^{P^s}$ from multiple antennas (i.e., different subscript $p$), we can obtain
\begin{align}\label{eq:muti-antenna}  
	\mathbf{R}_{\rm TS}^i = \mathbf{\Psi}\tilde{\mathbf{H}}^{{\rm eff}-i}_{\rm TS} + \tilde{\mathbf{W}}_{\rm TS}^i, \forall i,	
\end{align}
where $\mathbf{R}_{\rm TS}^i = \left[ \mathbf{r}_{{\rm TS},1}^{i}, \mathbf{r}_{{\rm TS},2}^{i} ,\dots, \mathbf{r}_{{\rm TS},P^s}^{i}  \right] \in \mathbb{C}^{G \times P^s}$, $\tilde{\mathbf{H}}^i_{\rm TS} = \left[ \tilde{\mathbf{h}}_{{\rm TS},1}^{{\rm eff}-i}, \tilde{\mathbf{h}}_{{\rm TS},2}^{{\rm eff}-i}, \dots, \tilde{\mathbf{h}}_{{\rm TS},P^s}^{{\rm eff}-i} \right] \in \mathbb{C}^{KL \times P^s}$, and $\tilde{\mathbf{W}}_{\rm TS}^i = \left[   \tilde{\mathbf{w}}_{{\rm TS},1}^{i}, \tilde{\mathbf{w}}_{{\rm TS},2}^{i}, \dots, \tilde{\mathbf{w}}_{{\rm TS},P^s}^{i} \right]  \in \mathbb{C}^{G \times P^s}$. 
Moreover, by stacking the received signals from multiple adjacent TSs (i.e., different superscript $i$), we can further obtain
\begin{align}\label{eq:muti-slots}  
	\mathbf{R}_{\rm TS} = \mathbf{\Psi} \tilde{\mathbf{H}}_{\rm TS}^{{\rm eff}} + \tilde{\mathbf{W}}_{\rm TS},	
\end{align}
where  $\mathbf{R}_{\rm TS} = \left[ \mathbf{R}_{\rm TS}^{(1)},\mathbf{R}_{\rm TS}^{(2)},\dots,\mathbf{R}_{\rm TS}^{(N+1)} \right]  \in \mathbb{C}^{G \times P^s(N+1)} $, $\tilde{\mathbf{H}}_{\rm TS}^{\rm eff} = \left[ \tilde{\mathbf{H}}^{{\rm eff}-(1)}_{\rm TS},\tilde{\mathbf{H}}^{{\rm eff}-(2)}_{\rm TS},\dots,\tilde{\mathbf{H}}^{{\rm eff}-(N+1)}_{\rm TS} \right]  \\ \in \mathbb{C}^{KL \times P^s(N+1)}$, and
$\tilde{\mathbf{W}}_{\rm TS} = \left[ \tilde{\mathbf{W}}_{\rm TS}^{(1)},\tilde{\mathbf{W}}_{\rm TS}^{(2)}, \dots,  \tilde{\mathbf{W}}_{\rm TS}^{(N+1)}\right]  \in \mathbb{C}^{G \times P^s(N+1)}$. 
The common sparsity pattern of $\tilde{\mathbf{H}}_{\rm TS}^{{\rm eff}}$ is illustrated in Fig. \ref{fig:common sparsity}.

\begin{figure}[t]	
	\centering
	\includegraphics[width=\columnwidth, keepaspectratio]{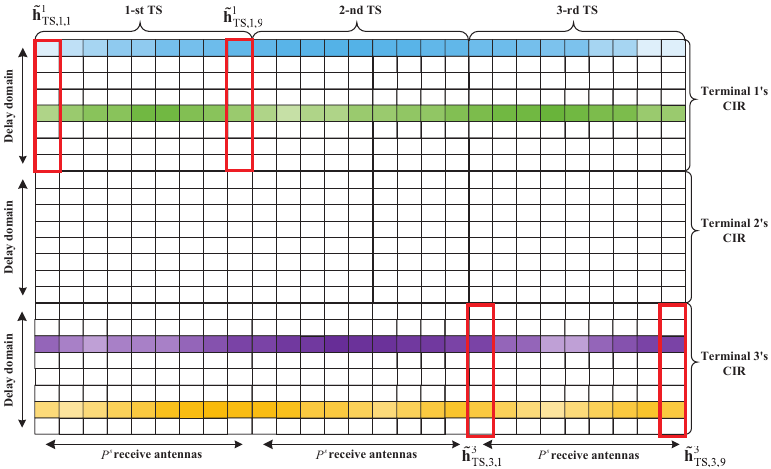}
	\captionsetup{font={footnotesize,color = {black}}, justification = raggedright,labelsep=period}
	\caption{The illustration of the common sparsity of $\tilde{\mathbf{H}}_{\rm TS}^{{\rm eff}}$ resulting from the spatial and temporal correlations in the TSLs. As an example, we assume that $K=3$ potential terminals access one LEO satellite, which is equipped with a $(P_x^s,P_y^s) = (3,3)$ UPA. Besides, the dimension of Doppler domain for OTFS waveform is $N=2$, and the maximum RToA and MPCs' delay are $L-1=7$.}
	\label{fig:common sparsity}	
\end{figure}

\subsection{Joint ATI and Coarse CE Based on MMV-CS Theory}\label{S4.2} 

Generally speaking, the dimension of non-ISI region $G$ is expected to be as small as possible to reduce the TSs overhead, so that $G$ could be usually far smaller than the dimension of $\tilde{\mathbf{H}}_{\rm TS}^{\rm eff}$. 
Nevertheless, based on Eq. (\ref{eq:muti-slots}), estimating the high-dimensional $\tilde{\mathbf{H}}_{\rm TS}^{\rm eff}$ from the low-dimensional non-ISI region is difficult, and conventional LS and linear minimum mean square error (LMMSE) estimators would fail. 
Fortunately, the CS theory has proved that high-dimension signals can be accurately reconstructed by low-dimensional uncorrelated observations if the target signal is sparse or approximately sparse \cite{CS}. 
For the joint sparse signal recovery of Eq. (\ref{eq:muti-slots}), various signal recovery algorithms have been developed, which aim to exploit the inherent common sparstiy to jointly recover a set of sparse vectors for enhanced performance \cite{CS}. 
  
We propose to utilize the simultaneous orthogonal matching pursuit (SOMP) algorithm \cite{SOMP} for fully exploiting the spatial-temporal joint sparsity of the CIR and the sparse traffic behavior of terrestrial IoT terminals with relatively low computational complexity, which is listed in the stage 1 part of \textbf{Algorithm \ref{alg1}}. Specifically, line 3-line 7 heuristically find the most correlated atom in each iteration by calculating the correlation coefficients in step 3 and augment the support set of non-zero elements in line 4. According to the current support set, the locally optimal solution is calculated in line 5. Then the residual is updated in line 6 for the next iteration until the stop condition meets.

The estimated support set of $\hat{\tilde{\mathbf{H}}}_{\rm TS}^{\rm eff}$ is denoted as $\mathcal{I}$ and the individual index of support set divided for each IoT terminal can be denoted as   
$\Omega_{k} = \{ \omega_{k}^q |\omega_{k}^q\in \mathcal{I}, (k-1)L \le \omega_{k}^q <  kL \}$, where $\omega_{k}^q$ is the $q$-th $(1 \le q \le |\Omega_{k}|_c)$ element of the set $\Omega_{k}$.
On the basis of the estimated support set and CIR vectors, a channel gain-based activity identificator \cite{KML} is proposed for ATI as follows
\begin{align}\label{eq:UI}
\hat{\alpha}_k = \left\{ \begin{array}{ll}
	    1, &  \frac{1}{P(N+1)}\sum_{p}\sum_{l=(k-1)L+1}^{kL}| \hat{\tilde{\mathbf{H}}}_{{\rm TS}_{\left[ l,p \right]}}^{\rm eff} |^2 \ge \xi \\
	    0, &  \frac{1}{P(N+1)}\sum_{p}\sum_{l=(k-1)L+1}^{kL}| \hat{\tilde{\mathbf{H}}}_{{\rm TS}_{\left[l,p\right]}}^{\rm eff} |^2 < \xi \\
\end{array} \right.,
\end{align}
where $\xi = \beta\max\{ \frac{1}{P(N+1)}\sum_{p}\sum_{l=(k-1)L+1}^{kL}| \hat{\tilde{\mathbf{H}}}_{{\rm TS}_{\left[l,p\right]}}^{\rm eff} |^2, \forall k \}$ and $\beta=0.1$ is the threshold factor\footnote{If the channel gain of the $k$-th IoT terminal is decided to be above the threshold $\xi$, the $k$-th IoT terminal is declared to be active. And $\beta=0.1$ is an empirical value to minimize the identification error probability in Eq. (\ref{eq45}).}.
As a result, the ATS can be represented by $\hat{\mathcal{A}} = \{k|\hat\alpha_k = 1, 1 \le k \le K\}$ and the cardinality of $\hat{\mathcal{A}}$ is denoted as $\hat{K}_a = |\hat{\mathcal{A}}|_c$. 


\subsection{CE Refinement with Parametric Approach}\label{S4.3}
\SetAlgoNoLine
\SetAlFnt{\small}
\SetAlCapFnt{\normalsize}
\SetAlCapNameFnt{\normalsize}
\begin{algorithm}[tb]
	\renewcommand{\algorithmicrequire}{\textbf{Input:}}
	\renewcommand{\algorithmicensure}{\textbf{Output:}}
	\caption{Proposed two-stage successive ATI and CE}
	\label{alg1} 
	\begin{algorithmic}[1]
		\REQUIRE
		Measurement signals $\mathbf{R}_{\rm TS}$ and sensing matrix $\mathbf{\Psi}$.\\
		
		\ENSURE 
		Estimated activity indicator $\hat{\alpha}_k, \forall k$, 
		the ATS $\hat{\mathcal{A}}$, and the correspoding CIR $\hat{h}_{k,p}[\kappa,\ell],
		\hat{h}_{k,p}^{\rm DD}[\ell,\upsilon], k \in \hat{\mathcal{A}}, \forall p$;
		
		\textbf{Stage 1 (Joint ATI and coarse CE)}
		
		\STATE Initialize $t=1$, the residual $\mathbf{R}_0 = \mathbf{R}_{\rm TS}$, the index of support set $\mathcal{I} = \emptyset$, and define $T_{\rm max}$ as the maximum number of iterations and termination threshold $\pi_{\rm th}$;
		
		\REPEAT
		
		\STATE 
		$i^* = \mathop{\arg\max}_{i=0,\dots,KL-1} \sum_{k=1}^{(N+1)P} \left \vert \langle [\mathbf{R}_{t-1}]_{:,k},\boldsymbol{\psi}_i \rangle \right \vert$;
		
		\STATE 
		$\mathcal{I} = \mathcal{I} \cup \{ i^* \}$;
		
		\STATE 
		$\hat{\mathbf{H}}_{\rm TS}^{\rm temp}  = \boldsymbol{\Psi}^{\dagger}_{[:,\mathcal{I}]}  \mathbf{R}_{\rm TS} $;
		
		\STATE 
		$\mathbf{R}_{t} = \mathbf{R}_{\rm TS} - \boldsymbol{\Psi}_{[:,\mathcal{I}]} \hat{\mathbf{H}}_{\rm TS}^{\rm temp} $;
		
		\STATE $t=t+1$;
		
		\UNTIL{  $T > T_{\rm max}$ or $\Vert \mathbf{R}_t \Vert_{\rm F}^2 < PG(N+1)\pi_{\rm th}$ } 	
			
		\STATE 
		$\hat{\tilde{\mathbf{H}}}_{{\rm TS}_{[\mathcal{I}, :]}}^{\rm eff} = \hat{\mathbf{H}}_{\rm TS}^{\rm temp}$; 
		
		\STATE Compute the estimated indicator according to Eq. (\ref{eq:UI}) and obtain the ATS $\hat{\mathcal{A}}$;
		
		\textbf{Stage 2 (CE refinement)}

		\FOR {$k \in \hat{\mathcal{A}}$}
		
		\STATE Compute the estimate of the RToA and MPCs' delay according to Eq. (\ref{eq:delay estimation});
		
		\STATE {Estimate the Doppler shift according to Eq. (\ref{eq:start})-(\ref{eq:end});}

		\STATE Compute the effective channel coefficients according to Eq. (\ref{eq:geff}); 
		
		\STATE Refine the results of CE by reconstructing CIR according to Eq. (\ref{eq:CE1}) and Eq. (\ref{eq:CE2});		
		\ENDFOR			
	\end{algorithmic}
\end{algorithm}


From the above discussion in Section \ref{s} and the channel model in Eq. (\ref{eq:eff-DD}), it can be observed that the separability, stability, and sparsity of the DD domain channels maintain in the TSLs, which motivates us to leverage the parametric approach to acquire the accurate estimation of the DD domain channel parameters and further refine the CE results.     
Specifically, the remanent RToA among different terminals and MPCs' delay for each terminal's TSL can be acquired from the index of support set of $\hat{\tilde{\mathbf{H}}}_{\rm TS}^{\rm eff}$ as\footnote{Here, we no longer distinguish LoS and NLoS paths, and uniformly treat them as MPCs with different subscript $q$.} 
\begin{align}\label{eq:delay estimation}
	\hat{\ell}_{k}^q = \omega_{k}^q-(k-1)L, \, k \in \hat{\mathcal{A}}, \, 1 \le q \le |\Omega_k|_c.	
\end{align}
Besides, the acquired $N+1$ sampled values of the time-varying CIR from $N+1$ adjacent TSs
can facilitate the super-resolution estimation of the Doppler shift. 
Specifically, we can use the one-dimensional estimating signal parameters via rotational invariance techniques (ESPRIT) algorithm \cite{ESPRIT}, which is a class of harmonic analysis algorithms by exploiting the underlying rotational invariance among signal subspaces. 

For the convenience of the following Doppler estimation, we define a path-gain  matrix $\hat\Upsilon_{k}^{q^*} \in \mathbb{C}^{(N+1) \times P^s}$ for the MPC with maximum energy as follows
\begin{align}\label{eq:effective CIR matrix}
	\hat\Upsilon_{k}^{q^*} = \mathrm{mat} \left( \hat{\tilde{\mathbf{H}}}_{{\rm TS}_{[\omega_{k}^{q^*},:]}}^{\rm eff}
	, P^s, N+1 \right)^{\rm T}, \forall k,
\end{align}
where $q^* = \mathop{\arg\max}_{q}\Vert { \hat{\tilde{\mathbf{H}}  }^{\rm eff}_{{\rm TS}_{[\omega_{k}^{q},:]}}} \Vert_2^2$.
In fact, each column vector of  $\hat\Upsilon_{k}^{q^*}$ is composed of CIR associated with 
multiple TSs, and its different column vectors originate from different receive antennas.
For the CIR of different TSs, they are continuous observation samples for the same channel.
While for the CIR of different receive antennas, the Doppler shift is approximately identical and thus they can be regarded as multiple snapshots to mitigate the effects of noise.
The main steps of Doppler estimation based on the ESPRIT algorithm are detailed as follows.

First of all, we divide two subarrays for each snapshot, which consist of CIR from the first $N$ TSs and the last $N$ TSs, respectively, as 
\begin{align} \label{eq:start}
	\mathbf{x}_{k,p}^1 =  \hat{\Upsilon}^{q^*}_{{k}_{ \left[ 1:N,p \right] }},
	\mathbf{x}_{k,p}^2 =  \hat{\Upsilon}^{q^*}_{{k}_{\left[ 2:N+1,p \right] }}, \forall k,p,
\end{align}	
and their combination $\mathbf{x}_{k,p} = \left[ \mathbf{x}_{k,p}^{1^{\rm T}} \; \mathbf{x}_{k,p}^{2^{\rm T}} \right]^{\rm T} \in  \mathbb{C}^{2N \times 1}$. In the presence of noise, the low rank property of autocorrelation matrix
\begin{align}
	{\mathbf{R}}_{xx}^k = E \left[\mathbf{x}_{k,p} \mathbf{x}_{k,p}^{\rm H} \right] \approx \frac{1}{P}\sum_{p=1}^{P}  \mathbf{x}_{k,p} \mathbf{x}_{k,p}^{\rm H},
\end{align} 
is destroyed. To mitigate the impact of noise, the eigenvalue decomposition (EVD) is utilized to distinguish the signal subspace and noise subspace, and we take the minimum eigenvalue $\hat{\sigma}^2_k$ as the estimate of the noise's variance. As a result, the noise cancelled autocorrelation matrix $\hat{\mathbf{R}}_{xx}^k$ can be calculated as 
\begin{align}
	\hat{\mathbf{R}}_{xx}^k = \mathbf{R}_{xx}^k-\hat{\sigma}^2_k \mathbf{I},
\end{align}
Then, the subspace of subarray $\mathbf{x}_{k,p}^1$ and $\mathbf{x}_{k,p}^2$ can be obtained 
by performing EVD on $\hat{\mathbf{R}}_{xx}^k$  as
\begin{align}
	\hat{\mathbf{R}}_{xx}^k =  \hat{\mathbf{U}}_k \hat{\mathbf{\Sigma}}_k \hat{\mathbf{U}}^{\rm H}_k,
\end{align}
and the first column of the eigenvector matrix  $\hat{\mathbf{U}}_k^s$ can approximate their  dominant signal subspace, i.e.,   $\mathbf{e}_{k}^1 = \hat{\mathbf{U}}_{k_{[1:N,1]}}^s$,  
$\mathbf{e}_k^2 = \hat{\mathbf{U}}_{k_{[N+1:2N,1]}}^s$.
In fact, $\mathbf{e}_{k}^1$ and $\mathbf{e}_{k}^2$ are characterized by rotational invariance \cite{ESPRIT}. Therefore, based on the LS criterion, the estimated Doppler shift can be calculated by
\begin{align} \label{eq:end}
	\hat\upsilon_{k} = \frac{N}{2\pi} \arg({{\mathbf{e}^1}^{\dagger}_{k}}\mathbf{e}_{k}^2).
\end{align}
Before proceeding with the estimation of channel coefficients of the MPCs, we introduce a lemma as follows.
\begin{lemma} \label{lemma1}	
We assume that the support set of $\hat{\tilde{\mathbf{H}}}_{\rm TS}^{\rm eff}$ is estimated perfectly. The non-zero elements of the recoverd sparse CIR vector associated with the $i$-th TS and $p$-th receive antenna are defined as $\hat{\mathbf{h}}^{{\rm eff-nz},i}_{{\rm TS}, p} =  \hat{\tilde{\mathbf{H}}}_{{\rm TS}_{[{\mathcal{I}},p+(i-1)P^s]}}^{\rm eff} \in \mathbb{C}^{Q \times 1}$, 
and the effective channel coefficients of the LoS and NLoS paths are denoted as 
\begin{align}\label{eq:geff-L}
	g_{k,p}^{\rm eff-LoS} = \sqrt{\frac{\gamma_k P_k}{\gamma_k+1}} g_k^{\rm LoS}  g_k^{\rm ABF-LoS}
	\left[ \mathbf{v}_R (\theta_{k}^{\rm zen}, \theta_{k}^{\rm azi}) \right]_p,\\
	g_{k,p}^{{\rm eff}-q} = \sqrt{\frac{P_k}{\gamma_k+1}} g_k^q g_k^{{\rm ABF}-q}  \left[ \mathbf{v}_R (\theta_{k}^{\rm zen}, \theta_{k}^{\rm azi}) \right]_p.
\end{align}
Their relationship can be written as      
\begin{align}\label{eq:g}
	\hat{\mathbf{h}}^{{\rm eff-nz},i}_{{\rm TS}, p}  = \boldsymbol{\Psi}_{[:,\mathcal{I}]}^{\dagger}  \mathbf{\Gamma}  \boldsymbol{\eta}^{i-1} \odot \mathbf{g}_{p}^{\rm eff} + \boldsymbol{\Psi}_{[:,\mathcal{I}]}^{\dagger} \mathbf{w}_{{\rm TS},p}^{i}, 
\end{align}
where $\mathbf{\Gamma}$ and $\boldsymbol{\eta}^{i-1}$are parameterized by the TSs, Doppler shift, RToA, and MPCs' delay, and $\mathbf{g}^{\rm eff}_p$ collects the effective channel coefficients of all active terminals associated with the $p$-th receive antenna. Their specific expression can refer to the Appendix.
\end{lemma}
\begin{proof}
 Please refer to Appendix.	
\end{proof}

After reconstructing $\hat{\mathbf{\Gamma}}$ and $\hat{\boldsymbol{\eta}}^{i-1}$ with the estimated Doppler shift $\hat\upsilon_{k}$, RToA, and MPCs' delay $\hat{\ell}_{k}^q$, the effective channel coefficients related to the $p$-th receive antenna can be mathematically calculated in line with the {Lemma \ref{lemma1}} as
\begin{align}\label{eq:geff}
	\hat{\mathbf{g}}_{p}^{\rm eff} = \frac{1}{N+1} \sum_{i=1}^{N+1} \left[ (\boldsymbol{\Psi}_{[:,\mathcal{I}]}^{\dagger}  \hat{\mathbf{\Gamma}}  )^{-1} {\hat{\mathbf{h}}^{{\rm eff-nz},i}_{{\rm TS}, p} \odot \frac{1}{\hat{\boldsymbol{\eta}}^{i-1}} } \right].		
\end{align}	


Up to now, the dominated DD domain channel's parameters have been acquired, and the results of CE refinement can be expressed as
\begin{align}\label{eq:CE1}  
	\hat{h}_{k,p}^{\rm eff-DD}[\ell, \upsilon] = \sum_{q=1}^{|{\Omega}_{k}|_c} \hat{g}_{k,p}^{{\rm eff}-q}  \delta[\ell-\hat{\ell}_{k}^q] \delta[\upsilon-\hat{\upsilon}_{k}]  , \forall k,p,
\end{align}
where $\hat{g}_{k,p}^{{\rm eff}-q}$ is the element of $\hat{\mathbf{g}}_{p}^{\rm eff}$ related to the $q$-th path and the $k$-th terminal. Besides, the estimate of time-varying CIR can be represented by 
\begin{align}\label{eq:CE2}   
	\hat{h}_{k,p}^{\rm eff}[\kappa,\ell] = \sum_{q=1}^{|{\Omega}_{k}|_c} \hat{g}_{k,p}^{{\rm eff}-q} e^{j2\pi\frac{\hat\upsilon_{k}(\kappa-\hat{\ell}_{k}^q)}{N(M+M_t)}}  \delta[\ell-\hat{\ell}_{k}^q], \forall k,p.
\end{align}

\begin{remark}\label{R2}
	In fact, ignoring the noise term, Eq. (\ref{eq:g}) can be rewritten as
	\begin{align}
		\hat{\mathbf{h}}^{{\rm eff-nz},i}_{{\rm TS}, p} \approx \boldsymbol{\zeta} \odot \boldsymbol{\eta}(i,\{\upsilon_k\}_{k \in \mathcal{A}}), 
	\end{align}
	where $\boldsymbol{\zeta} = \boldsymbol{\Psi}_{[:,\mathcal{I}]}^{\dagger}  \mathbf{\Gamma} \mathbf{g}_{p}^{\rm eff} \in \mathbb{C}^{Q \times 1}$ is a Doppler-invariant and $i$-invariant vector, and $\mathbf{\eta}$ is a function of $i$ and Doppler shift $\{\upsilon_k\}_{k \in \mathcal{A}}$. The linear relationship between $\boldsymbol{\zeta}$ and $\boldsymbol{\eta}$ ensures that the proposed Doppler shift estimation method is immune from the uncertainty of $\boldsymbol{\zeta}$, i.e., the approximation error in Eq. (\ref{eq:approx vec}), which also guarantees the effectiveness of the subsequent channel coefficients estimation. 
\end{remark}

So far, we have finished the discussion of the proposed two-stage successive ATI and CE scheme, and its complete procedure is listed in \textbf{Algorithm \ref{alg1}}.

\subsection{Complexity Analysis}

The computational complexity of the proposed two-stage successive ATI and CE scheme in \textbf{Algorithm \ref{alg1}}
mainly depends on the following operations:

\begin{itemize}
	\item \textbf{SOMP} \textit{(line 1-9)}:
	The matrix-vector multiplications involved  in each iteration have the complexity on the order of
	$\mathcal{O} \left( 2G \left(N+1 \right)P^sKL + KL \left( N+1 \right)P^s + 2t^2G   \right.$  $\left. + t^3 + tG \left(N+1 \right)P^s \right)$, where $t$ is the iteration index.
	
	\item \textbf{Activity identificator} \textit{(line 10)}: The cost to identify the active terminals is $\mathcal{O} \left( P^s \left( N+1 \right)KL \right)$.
	
	\item  \textbf{Doppler estimation} \textit{(line 13)}: The computational complexity of Doppler estimation is $\mathcal{O} \left( \left( 2NP^s + (2N)^3 + N \right) \hat{K}_a \right)$.
	
	\item  \textbf{Effective channel coefficients estimation} \textit{(line 14)}: The computational complexity involved mainly results from LS solution, where it has the complexity on the order of 
	$\mathcal{O} \left(  P^s \left[ \left(\sum_{k \in \hat{\mathcal{A}}} 1 + \hat{Q}_k\right)^2G + \left({\sum_{k \in \hat{\mathcal{A}}} 1+\hat{Q}_k}\right)^3 +  \right. \right. $ 
	$\left.\left.  2 \left( {\sum_{k \in \hat{\mathcal{A}}} 1 + \hat{Q}_k } \right) \right]  \right)$.	
\end{itemize}

Obviously the SOMP implemented in Algorithm 1 contributes to most of the computational complexity. In comparison of the proposed algorithm, the complexity of the TDSBL-FM algorithm \cite{GF-NOMA-OTFS}
is cubic to the pilot length and quadratic to number of potential terminals and also requires a good deal of iterations to converge, which is prohibitive high for large-scale system. Meanwhile, the proposed scheme has the same order of computational complexity as the 3D-SOMP in \cite{SWQ.OTFS}, but requires fewer memory resources and computing time overhead without vectorization operation.

\color{black}
\section{Signal Detection}\label{S5}

Based on the above ATI and CE results, we develop a LS-based parallel time domain multi-user SD for detecting the OTFS signals with relatively low computational complexity in this section.

\subsection{Received Signal Preprocessing}\label{S5.1}

As illustrated in Fig.~\ref{fig3:data frame}, on one hand, the received OTFS payload data symbols over the time dispersive channels would be contaminated by trailing of the preceding TS; on the other hand, the data symbols can also contaminate the following TS part. 
These can lead to severe ISI.    
Fortunately, TSs are known by the transceiver.
With the estimated CSI, the aforementioned ISI can be eliminated to facilitate the following SD. 
The details of the preprocessing of the received signals before SD are presented as follows.   

First, Eq. (\ref{eq:IO}) can be rewritten as a  vector form as 
\begin{align}\label{eq35}
	\mathbf{r}_p = \sum_{k=1}^K  \alpha_k  \sqrt{P}_k \underbrace{\left( \mathbf{\Pi}_{k,p}^{{\rm eff-LoS}} + \sum_{q=1}^{Q_k} \mathbf{\Pi}_{k,p}^{{\rm eff}-q} \right) }_{\mathbf{\Pi}_{k,p}^{\rm eff}} \mathbf{s}_k + \mathbf{w}_p,\; \forall p,
\end{align}
where $\mathbf{\Pi}_{k,p}^{\rm eff-LoS} \, (\mathrm{also} \; \mathbf{\Pi}_{k,p}^{{\rm eff}-q}) \in \mathbb{C}^{(MN+M_tN+M_t) \times (MN+M_tN+M_t)}$ consists of elements of time-varying CIR in Eq. (\ref{eq:T-CIR}), and its $(m, n)$-th element is defined as    
\begin{align}\label{eq:channel matrix}
	 \mathbf{\Pi}^{{\rm eff-LoS}}_{k,p_{\left[ m,n \right]}} = \left\{ \begin{array}{ll}
		h_{k,p}^{\rm eff}[m-1, m-n], & m-n=\ell_k^{\rm LoS} \\
		0,  & \textrm{otherwise}\\
	\end{array} \right.,
\end{align}
and $\mathbf{\Pi}^{{\rm eff}-q}_{k,p}$ shares a similar expression.
Consequently, the ISI in the OTFS payload data symbol caused by the trailing of the preceding TS can be estimated as  
\begin{align}\label{eq:SIC}
	\hat{\mathbf{r}}_p^{\rm ISI} = \sum_{k = 1}^K \hat{\alpha}_{k}  \hat{\mathbf{\Pi}}_{k,p}^{\rm eff} \mathring{\mathbf{s}}_k, \; \forall p,
\end{align}
where $\mathring{\mathbf{s}}_k$ consists of TSs and zero sequences, i.e., $\mathring{\mathbf{s}}_k = \left[ \mathbf{c}_k^{\rm T},\mathbf{0}_{M \times 1}^{\rm T},\mathbf{c}_k^{\rm T},\mathbf{0}_{M \times 1}^{\rm T},\dots,\mathbf{c}_k^{\rm T},\mathbf{0}_{M \times 1}^{\rm T},\mathbf{c}_k^{\rm T} \right]^{\rm T}  \in \mathbb{C}^{(M_tN+MN+M_t) \times 1}$, and $\hat{\mathbf{\Pi}}_{k,p}^{\rm eff}$ is the estimate of $\mathbf{\Pi}_{k,p}^{\rm eff}$ with its elements padded by the estimated CIR in Eq. (\ref{eq:CE2}). 

Besides, to form the cyclic convolution relationship between the OTFS data signal and the CIR like traditional OFDM/OFDM-based OTFS system \cite{OFDM-based OTFS}, the data trailing (cause the ISI to the following TS) will be superposed onto the header of each OTFS data symbol region.
In fact, such a data trailing can be acquired and shifted by 
\begin{align}\label{eq:compensate}
	\hat{\mathbf{r}}_p^{\rm tra}  = \mathbf{R}_{t}^{\rm T} (\mathbf{I}_{N} \otimes \mathbf{R}_{s}) \mathbf{R}_{t} (  {\mathbf{r}}_p - {\mathbf{r}}_p^{\rm ISI}), \; \forall p, 
\end{align}	
where $\mathbf{R}_{t} = \left[ \mathbf{0}_{(M+M_t)N  \times M_t} \; \mathbf{I}_{(M+M_t)N} \right] \in \mathbb{C}^{(M+M_t)N \times (MN+M_tN+M_t)}$ and
the $(m, n)$-th element of  $\mathbf{R}_s \in \mathbb{C}^{(M+M_t) \times (M+M_t)}$ is defined as  
\begin{align}\label{eq:superposition matrix}
	 \mathbf{R}_{s_{[m,n]}}	= \left\{ \begin{array}{ll}
		1 & m-n = M, \, n \in [1,M_t]\\
		0 & \textrm{otherwise}\\
	\end{array} \right..
\end{align} 
Therefore, the  preprocessed OTFS payload data symbols $\hat{\mathbf{r}}_p \in \mathbb{C}^{MN \times 1}$ can be finally acquired after removing the TSs as 
\begin{align}\label{eq:TS removal}
	\hat{\mathbf{r}}_p  = (\mathbf{I}_{N} \otimes \mathbf{R}_{r}) \mathbf{R}_{t} ({\mathbf{r}}_p - \hat{\mathbf{r}}_p^{\rm ISI} + \hat{\mathbf{r}}_p^{\rm tra}) , \; \forall p,
\end{align}	 
where $\mathbf{R}_{r} = \left[ \mathbf{I}_{M} \; \mathbf{0}_{M \times M_t} \right] \in \mathbb{C}^{M \times (M+M_t)}$.

%
%
%
%
%
%
%
%
%
%

\subsection{LS-Based Parallel Time Domain Multi-User SD}\label{S5.2}
  
It has been shown that with the fractional Doppler shift, the TF domain and DD domain effective channel matrices could be not very sparse due to the Doppler spreading with the limited Doppler resolution, while the sparsity of time domain channel remains to hold \cite{yjh}. 
Therefore, we are motivated to perform multi-user SD in the time domain to exploit its sparse  pattern for lower computational complexity.
In fact, based on Eq. (\ref{eq35}) and Eq. (\ref{eq:SIC}), we have  
\begin{align}\label{eq:key}
	& \mathbf{r}_p -  \mathbf{r}_p^{\rm ISI}   =   \sum_{k \in  \hat{\mathcal{A} }} \hat{\mathbf{\Pi}}_{k,p}^{\rm eff} \left( \mathbf{s}_k - \mathring{\mathbf{s}}_k \right) + 
	\sum_{k \in {  \hat{\mathcal{A}} }} \left(\sqrt{P}_k {\mathbf{\Pi}}_{k,p}^{\rm eff} - \hat{\mathbf{\Pi}}_{k,p}^{\rm eff} \right) \mathbf{s}_k \notag
	\\ & \quad\quad  + \sum_{k \in ({ \mathcal{A} - \hat{\mathcal{A}}) }} \sqrt{P}_k { \mathbf{\Pi}}_{k,p}^{\rm eff} \mathbf{s}_k  + \mathbf{w}_p 
	 =  \sum_{k \in { \hat{\mathcal{A}} }} \hat{\mathbf{\Pi}}_{k,p}^{\rm eff} \left( \mathbf{s}_k - \mathring{\mathbf{s}}_k \right) + \hat{\mathbf{w}}_p,
\end{align}
where $\hat{\mathbf{w}}_p =  \sum_{k \in {   \hat{\mathcal{A}} }} 
\left( \sqrt{P_k} {\mathbf{\Pi}}_{k,p}^{\rm eff} - \hat{\mathbf{\Pi}}_{k,p}^{\rm eff} \right) \mathbf{s}_k + \sum_{k \in ({ \mathcal{A} - \hat{\mathcal{A}}) }} \sqrt{P_k} {\mathbf{\Pi}}_{k,p}^{\rm eff} \mathbf{s}_k + \mathbf{w}_p$ is the effective noise vector including errors in the signal preprocessing and AWGN.
Furthermore, by stacking Eq. (\ref{eq:compensate}), Eq. (\ref{eq:TS removal}), and Eq. (\ref{eq:key}), we have 
\begin{align}\label{eq: IO}
	\hat{\mathbf{r}}_p = & \sum_{k \in  \hat{\mathcal{A}}} 
	\left[  \mathbf{I}_N \otimes \mathbf{R}_r (\mathbf{I}_{M+M_t}+\mathbf{R}_s)   \right] \mathbf{R}_t \left(  \mathbf{r}_p -  \hat{\mathbf{r}}_p^{\rm ISI} \right)  \nonumber \\
	= & \sum_{k \in  \hat{\mathcal{A}} } 
	\left[  \mathbf{I}_N \otimes \mathbf{R}_r (\mathbf{I}_{M+M_t}+\mathbf{R}_s)   \right] \mathbf{R}_t \hat{\mathbf{\Pi}}_{k,p}^{\rm eff} \mathbf{R}_t^{\rm T} ( \mathbf{I}_N \otimes \mathbf{A}_t)   \tilde{\mathbf{s}}_k \notag \\
	& \quad\quad + \left[  \mathbf{I}_N \otimes \mathbf{R}_r    (\mathbf{I}_{M+M_t}+\mathbf{R}_s)   \right] \mathbf{R}_t \hat{\mathbf{w}}_p,
\end{align}
where $\mathbf{A}_{t} = \left[ \mathbf{0}_{M \times M_t}^{\rm T} \; \mathbf{I}_{M}^{\rm T} \right]^{\rm T} \in \mathbb{C}^{(M+M_t) \times M}$.

With the aid of the TSs and the preprocessing aforementioned, the ISI between  adjacent  OTFS data symbols can be avoided. 
Hence, the SD in the time domain can be performed in parallel for $N$ OTFS payload data symbols, which can significantly reduce the computational complexity. 
As a result, Eq. (\ref{eq: IO}) can be further decomposed into  
\begin{align} \label{eq:single IO}
	\hat{\mathbf{r}}_p^i & = \sum_{k \in {\hat{\mathcal{A}}}} \underbrace{ \mathbf{R}_{r}  ( \mathbf{I}_{M+M_t} +\mathbf{R}_{s} ) \hat{\mathbf{\Pi}}_{k,p}^{{\rm eff}-i}  \mathbf{A}_t }_{\hat{\mathbf{U}}_{p,k}^i} \tilde{{\mathbf{s}}}_k^i + \hat{\mathbf{w}}_p^{i},\, \forall i,p,
\end{align}
where  $\hat{\mathbf{\Pi}}_{k,p}^{{\rm eff}-i} = \left(\mathbf{R}_t \hat{\mathbf{\Pi}}_{k,p}^{\rm eff} \mathbf{R}_{t}^{T} \right)_{[(i-1)(M+M_t)+1:i(M+M_t)]} \in \mathbb{C}^{(M+M_t) \times (M+M_t)}$, $\hat{\mathbf{r}}_p^i = \hat{\mathbf{r}}_{p_{[(i-1)M+1:iM]}} \in \mathbb{C}^{M \times 1}$,
and $\hat{\mathbf{w}}_p^{i} \in \mathbb{C}^{M \times 1}$ is the corresponding noise vector.

Moreover, we intend to extend Eq. (\ref{eq:single IO}) to jointly process the received signal from $P^s$ receive antennas as
\begin{align}\label{eq:final IO}
	\hat{\mathbf{r}}^i = \hat{\mathbf{U}}^i \tilde{\mathbf{s}}^{i} + \hat{\mathbf{w}}^{i}, \, \forall i,
\end{align}
where $\hat{\mathbf{U}}^i \in \mathbb{C}^{P^sM \times \hat{K}_aM}$ is a block matrix and its $(p,k)$-th submatrix equals to $\hat{\mathbf{U}}^i_{p,k}$,
$\hat{\mathbf{r}}^i = \left[ \hat{\mathbf{r}}_1^{i^{\rm T}},\hat{\mathbf{r}}_2^{i^{\rm T}},\dots,\hat{\mathbf{r}}_{P^s}^{i^{\rm T}} \right]^{\rm T} \in \mathbb{C}^{MP^s \times 1}$, 
$\tilde{\mathbf{s}}^i = \left[ \tilde{\mathbf{s}}_{k_1}^{i^{\rm T}},
\tilde{\mathbf{s}}_{k_2}^{i^{\rm T}},\dots, \right.$ $\left. \tilde{\mathbf{s}}_{k_{\hat{K}_a}}^{i^{\rm T}} \right]^{\rm T} \in \mathbb{C}^{M \hat{K}_a \times 1} $, $k_1,k_2,\dots,k_{\hat{K}_a}$ are the elements of the set $\hat{\mathcal{A}}$, and $\hat{\mathbf{w}}^{i}$ denotes the noise vectors of different receive antennas.

Therefore, given $P \ge \hat{K}_a$, the time domain OTFS signals of different active terminals can be detected by calculating the LS solution of Eq. (\ref{eq:final IO}).
Benefitting from the sparsity of TSLs, $\hat{\mathbf{U}}^i$ displays favorable sparse pattern, where we can further utilize iterative method, such as LS QR-factorization (LSQR) \cite{lsqr}, to facilitate the approximate solution of sparse linear equations with low  computational complexity.
In this case, the computational complexity of Eq. (\ref{eq:final IO}) solution dramatically reduces from $\mathcal{O} \left(  2P^sM^3 \hat{K}_a^2 + \hat{K}_a^3 M^3 + \hat{K}_a M  \right)$ to $\mathcal{O} \left( 2P^sMQT \right)$, where $T$ is the required number of iterations for LSQR algorithm.

\color{black}
\section{Performance Evaluation}\label{S6}

\subsection{Simulation Setup}

In this section, we carry out extensive simulation investigations to verify
the effectiveness of our proposed scheme under different parameter configurations, and compare it with the state-of-the-art solutions.  
First of all, we define the identification error probability $P_e$ for ATI as
\begin{align}\label{eq45}
P_e =  \sum_{k=1}^{K} |\hat{\alpha}_k-\alpha_k|. 	
\end{align}
Besides, the normalized mean square error (NMSE) for CE is considered as
\begin{align}\label{eq46} 
\mathrm{NMSE} = \frac{\sum_{k=1}^K \sum_{p=1}^P ||\hat{\alpha}_k \hat{\mathbf{\Pi}}_{k,p}^{\rm eff} - \alpha_k \sqrt{P_k} \mathbf{\Pi}_{k,p}^{\rm eff} ||_{\rm F}^2}{\sum_{k=1}^K \sum_{p=1}^P ||\alpha_k \sqrt{P}_k \mathbf{\Pi}_{k,p}^{\rm eff} ||_{\rm F}^2},	
\end{align}	
and the uncoded bit error rate (BER) for SD is considered as
\begin{align}\label{eq47} 
\mathrm{BER} = \frac{E_aNMM_b+B_a}{K_aNMM_b}, 	
\end{align}
where $E_a$ is the number of falsely identified terminals and $M_b$ is the modulation order. 
Besides, $B_a$ is the total error bits for the DD domain payload data of the correctly identified active terminals, i.e., 
error bits between $\hat{\mathbf{X}}_k^{\rm DD}$ and ${\mathbf{X}}_k^{\rm DD}, k \in (\mathcal{A} \cap \hat{\mathcal{A}})$.

\begin{table}[]
	\centering
	\caption{Simulation parameters}
	\label{tab:my-table2}
	\resizebox{0.9\columnwidth}{!}{%
		\begin{tabular}{|c|c|c|}
			\hline
			\textbf{Contents}                & \textbf{Parameters}                                                                                        & \textbf{Values}                \\ \hline
			\multirow{8}{*}{\textbf{System}} & Carrier frequency                                                                                          & 10 GHz                         \\ \cline{2-3} 
			& Subcarrier spacing                                                                                         & 480 KHz                        \\ \cline{2-3} 
			& Bandwidth                                                                                                  & 122.88 MHz                     \\ \cline{2-3} 
			& OTFS data size $(M,N)$                                                                                     & (256,8)                        \\ \cline{2-3} 
			& Modulation scheme                                                                                          & QPSK                           \\ \cline{2-3} 
			& Satellite's UPA $(P_x^s, P_y^s)$                                                                           & $(32,32)$                        \\ \cline{2-3} 
			& Terminals' UPA $(P_x^t, P_y^t)$                                                                            & $(32,32)$                        \\ \cline{2-3} 
			& Angular spacing $(\Delta_z, \Delta_a)$                                                                     & $(14.4^{\circ},14.3^{\circ})$  \\ \hline
			\multirow{8}{*}{\textbf{TSL}}    & Orbit altitude of LEO satellite                                                                            & 500 km                         \\ \cline{2-3} 
			& Velocity of LEO satellite                                                                                  & 7.58 km/s                      \\ \cline{2-3} 
			& Velocity of terminals                                                                                      & $0 \sim 10$ m/s                \\ \cline{2-3} 
			& Service coverage radius                                                                                    & 494.8 km                       \\ \cline{2-3} 
			& Zenith angle $\theta_k^{\rm zen}$                                                                          & $[-44.7^{\circ},44.7^{\circ}]$ \\ \cline{2-3} 
			& Azimuth angle $\theta_k^{\rm azi}$                                                                         & $[0,360^{\circ})$              \\ \cline{2-3} 
			& \begin{tabular}[c]{@{}c@{}}Remant RToA $\tau_k^{\rm LoS}$ and \\  MPC's delay $\tau_k^q$ range\end{tabular} & $0 \sim 0.52 \, \mu \rm s$              \\ \cline{2-3} 
			& Doppler shift $\nu_k^{\rm LoS} \, (\nu_k^q)$ range                                                       & $0 \sim 178.2 \, \rm{KHz}$     \\ \hline
		\end{tabular}%
	}
\end{table}

For the massive MIMO based massive connectivity in IoT, the performance of payload data demodulation highly depends on the channel statistics of the simultaneously served terminals. 
As for the LEO satellite-based IoT, the channel characteristics are mainly determined by the simultaneously served terminals' AoA observed at the satellite receiver, i.e., $\theta_k^{\rm zen}$ and $\theta_k^{\rm azi}, \forall k \in \mathcal{A}$ \cite{You.LEO}. 
Therefore, channels of IoT terminals with minor AoA differences would be highly correlated, which inevitably leads the multi-user MIMO channel matrix to be ill-conditioned with considerably performance deterioration if they are allocated with the same TF or DD resources. 
To this end, we adopt a user retransmission strategy to eliminate interference in the presence of correlated channels. Specifically, when the LEO satellite identifies that the zenith and azimuth angles of terminals meet 
\begin{align}
	|\theta_{k_1}^{\rm zen} - \theta_{k_2}^{\rm zen}| < \Delta_z, \;
	|\theta_{k_1}^{\rm azi} - \theta_{k_2}^{\rm azi}| < \Delta_a, 
\end{align}
it broadcasts retransmission scheduling signaling for one of the
collided terminals $k^* = k_1 \ (k_2)$ to retransmit its signal, 
where $k_1, k_2 \in \mathcal{A}$ and $k_1 \neq k_2$, $\Delta_z$ and $\Delta_e$ are the preset minimum zenith and azimuth spacing to avoid the MIMO channel matrix to be ill-conditioned \cite{You.LEO}. 
After receiving the retransmission signal $\mathbf{r}_p^2, \forall p$,
it is necessary for the LEO satellite to reestimate the CSI with the TSs, since the channel state may change or channel resource allocation could vary in the retansmission stage. 
Then, the LEO satellite could demodulate the payload data of the retransmission signal with the estimated CSI as $\hat{\tilde{\mathbf{S}}}_{k^*}$ as Eq. (\ref{eq:SPC}).
In this way, the LEO satellite can eliminate the interference resulting from terminal $k^*$ for the first stage received signal $\mathbf{r}_p$ as 
\begin{align}
	\mathbf{r}_p^{1-2} = \mathbf{r}_p - \hat{\boldsymbol{\Pi}}_{k^*,p}^{\rm eff} \hat{\mathbf{s}}_{k^*},
\end{align}
where $\hat{\boldsymbol{\Pi}}_{k^*,p}^{\rm eff}$ is the estimated CSI of terminal $k^*$ in the first stage, $\hat{\mathbf{s}}_{k^*}  = \left[ \mathbf{c}_{k^*}^{\rm T}, \tilde{\mathbf{s}}_{k^*}^{1^{\rm T}}  , \mathbf{c}_{k^*}^{\rm T}  , \tilde{\mathbf{s}}_{k^*}^{2^{\rm T}} , \dots , \mathbf{c}_{k^*}^{\rm T}  , \right. \\ \left. \tilde{\mathbf{s}}_{k^*}^{N^{\rm T}}  ,   \mathbf{c}_{k^*}^{\rm T}   \right]^{\rm T} \in \mathbb{C}^{(M_tN+MN+M_t) \times 1}$ is the reconstructed transmit signal of terminal $k^*$.
Following the procedures in Eq. (\ref{eq35})-(\ref{eq:final IO}), the satellite can demodulate the payload data for other terminals with $\mathbf{r}_p^{1-2}$.

\color{black}
Additionally, to meet the mutual coherence property (MCP) \cite{CS} of the sensing matrix for reliable recovery of sparse vectors in Eq. (\ref{eq:muti-slots}), we assume the time domain TS associated with the $k$-th terminal is generated from a standard complex Gaussian distribution, i.e., $\mathbf{c}_k \sim \mathcal{CN}(\mathbf{0}_{M_t \times 1}, \mathbf{I}_{M_t})$.
The system topology is set as Fig. \ref{fig1:system model}. 
Within the coverage radius of $494.8 \, \rm km$, the number of potential and active IoT terminals served by one LEO satellite is fixed as $K=100$ and $K_a=10$, respectively.
The zenith and azimuth angles of terminals are randomly distributed in the range of $[-44.7^\circ, 44.7^\circ]$ and $[0^\circ, 360^\circ)$, and their locations can be simulated accordingly. Besides, their velocity is randomly distributed in the range of $0 \sim 10$ m/s.
Other detailed system parameters for the following simulations are summarized in Table \ref{tab:my-table2}. 	
Meanwhile, the maximum number of iterations $T_{\rm max}$ in \textbf{Algorithm \ref{alg1}} is set to 30, and the termination threshold is set as $\pi_{\rm th} = 1.05\sigma_w^2$, where the variance $\sigma_w^2$ of AWGN at the receiver is assumed to be the \textit{a prior} information for the LEO satellite.


\begin{table}[t]
	\centering
	\captionsetup{font={color = {black}}, justification = raggedright,labelsep=period}
	\caption{Link budget for single terminal \cite{link budget 1}}
	\label{tab:my-table3}
	\resizebox{0.9\columnwidth}{!}{%
		\begin{tabular}{cccc}
			\hline
			\multicolumn{1}{|c|}{}                                      & \multicolumn{1}{c|}{\textbf{Case 1}} & \multicolumn{1}{c|}{\textbf{Case 2}} & \multicolumn{1}{c|}{\textbf{Case 3}} \\ \hline
			\multicolumn{1}{|c|}{Terminal zenith angle}                 & \multicolumn{1}{c|}{$0^{\circ}$}     & \multicolumn{1}{c|}{$25^{\circ}$}    & \multicolumn{1}{c|}{$44.7^{\circ}$}  \\ \hline
			\multicolumn{1}{|c|}{Frequency [GHz]}                       & \multicolumn{1}{c|}{10}              & \multicolumn{1}{c|}{10}              & \multicolumn{1}{c|}{10}              \\ \hline
			\multicolumn{1}{|c|}{Bandwidth [MHz]}                       & \multicolumn{1}{c|}{122.88 \cite{embb}}          & \multicolumn{1}{c|}{122.88 \cite{embb}}          & \multicolumn{1}{c|}{122.88 \cite{embb}}          \\ \hline
			\multicolumn{1}{|c|}{Tx transmit power [dBm]}               & \multicolumn{1}{c|}{40}              & \multicolumn{1}{c|}{40}              & \multicolumn{1}{c|}{40}              \\ \hline
			\multicolumn{1}{|c|}{Tx beamforming gain [dB]}              & \multicolumn{1}{c|}{40}              & \multicolumn{1}{c|}{40}              & \multicolumn{1}{c|}{40}              \\ \hline
			\multicolumn{1}{|c|}{Rx $G_r/T$ [dB/K]}                       & \multicolumn{1}{c|}{-4.62 \cite{link budget 2}}           & \multicolumn{1}{c|}{-4.62 \cite{link budget 2}}          & \multicolumn{1}{c|}{-4.62 \cite{link budget 2}}           \\ \hline
			\multicolumn{1}{|c|}{Free space path loss [dB]}             & \multicolumn{1}{c|}{167.25}          & \multicolumn{1}{c|}{168.10}          & \multicolumn{1}{c|}{170.21}          \\ \hline
			\multicolumn{1}{|c|}{Atmospheric loss [dB]}                 & \multicolumn{1}{c|}{0.07}            & \multicolumn{1}{c|}{0.07}            & \multicolumn{1}{c|}{0.07}            \\ \hline
			\multicolumn{1}{|c|}{Shadowing margin [dB]}                 & \multicolumn{1}{c|}{3}               & \multicolumn{1}{c|}{3}               & \multicolumn{1}{c|}{3}               \\ \hline
			\multicolumn{1}{|c|}{Scintillation loss [dB]}               & \multicolumn{1}{c|}{2.2}             & \multicolumn{1}{c|}{2.2}             & \multicolumn{1}{c|}{2.2}             \\ \hline
			\multicolumn{1}{|c|}{Polarization loss [dB]}                & \multicolumn{1}{c|}{0}               & \multicolumn{1}{c|}{0}               & \multicolumn{1}{c|}{0}               \\ \hline
			\multicolumn{1}{|c|}{Additional losses [dB]}                & \multicolumn{1}{c|}{0}               & \multicolumn{1}{c|}{0}               & \multicolumn{1}{c|}{0}               \\ \hline
			\multicolumn{1}{|c|}{Additional margin [dB]}                & \multicolumn{1}{c|}{6}               & \multicolumn{1}{c|}{6}               & \multicolumn{1}{c|}{6}               \\ \hline
			\multicolumn{1}{l}{}                                        & \multicolumn{1}{l}{}                 & \multicolumn{1}{l}{}                 & \multicolumn{1}{l}{}                 \\ \hline
			\multicolumn{1}{|c|}{\textbf{SNR for single terminal [dB]}} & \multicolumn{1}{c|}{14.59}            & \multicolumn{1}{c|}{13.73}            & \multicolumn{1}{c|}{11.62}            \\ \hline
		\end{tabular}%
	}
\end{table}

\subsection{Link budget}

According to the general formula for link budget derived from \cite{NB-IoT}, which includes all the gains and losses in the propagation medium from transmitter to receiver, the received signal-to-noise ratio (SNR) for single terminal case can be computed as follows 
\begin{align}\label{eq:snr}
	&	\textrm{SNR}_{\rm SU}\textrm{[dB]}  = P_k \textrm{[dBW]} + G_t\textrm{[dB]} + G_r/T\textrm{[dB]} \notag \\ 
	& \quad\quad\quad - k_{B} \textrm{[dBW/K/Hz]} 
		 - \textrm{PL}_{\rm FS} \textrm{[dB]}   
		- \textrm{PL}_{A} \textrm{[dB]} - \textrm{PL}_{S} \textrm{[dB]} \notag \\ 
	& \quad\quad\quad - \textrm{PL}_{\rm AD} \textrm{[dB]} 
		 -\textrm{PL}_{\rm MA} \textrm{[dB]} - 10\log_{10}(B_w{[\rm Hz]}),  
\end{align}
where $G_t$ is the beamforming gain for the terminal, $G_r$ and $T$ are the receive antenna gain and the noise temperature, respectively, and $k_B = -228.6 \ {\rm dBW/K/Hz}$ is the Boltzmann constant, $\textrm{PL}_{\rm FS}$ represents the free space propagation loss, $\textrm{PL}_{A}$ corresponds to atmospheric gas losses, $\textrm{PL}_{S}$ is a
shadowing margin, $\textrm{PL}_{\rm AD}$ denotes some underlying additional
losses due to the scintillation phenomena, $\textrm{PL}_{\rm MA}$ is the additional reserved margin, and $B_w$ is the channel bandwidth. 
TABLE \ref{tab:my-table3} illustrates the link budget for different configurations in the uplink RA. 
Unless otherwise mentioned, the transmit power is set to $P_k = 40 \ \rm dBm$ for all terminals in the following simulations.  

\begin{figure*}[tp]
	\centering
	\subfigure[]{	
		\includegraphics[width=0.66\columnwidth, keepaspectratio]{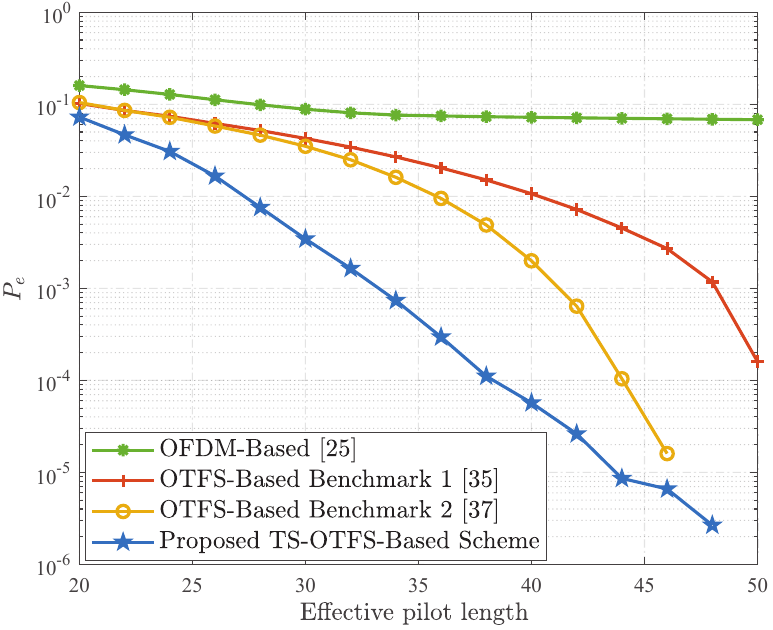}}	
	\subfigure[]{	
		\includegraphics[width=0.66\columnwidth, keepaspectratio]{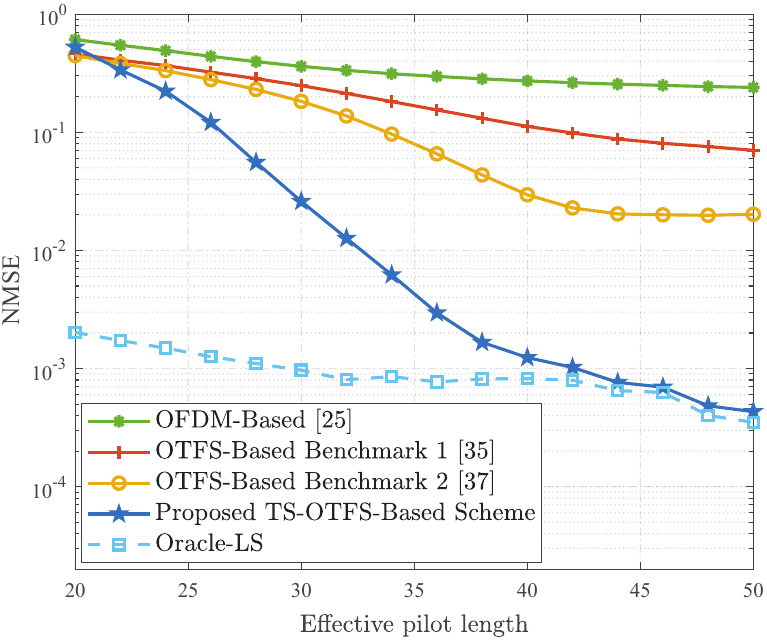}}	
	\subfigure[]{	
		\includegraphics[width=0.66\columnwidth, keepaspectratio]{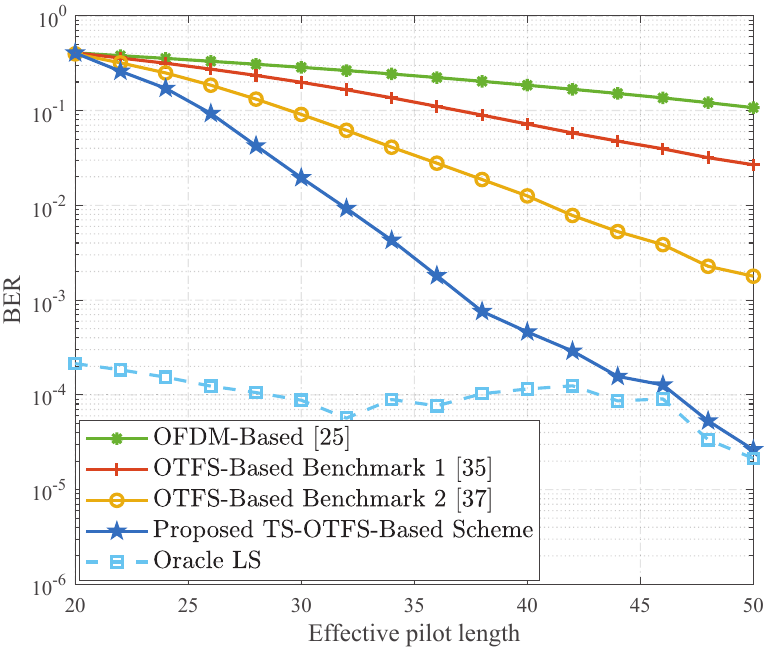}}	
	\caption{Performance comparison under different pilot length: (a) $P_e$; (b) NMSE; (c) BER.}
	\label{fig:non-ISI size}
\end{figure*}  

\begin{figure*}[tp]
	\centering	
	\subfigure[]{	
		\includegraphics[width=0.66\columnwidth, keepaspectratio]{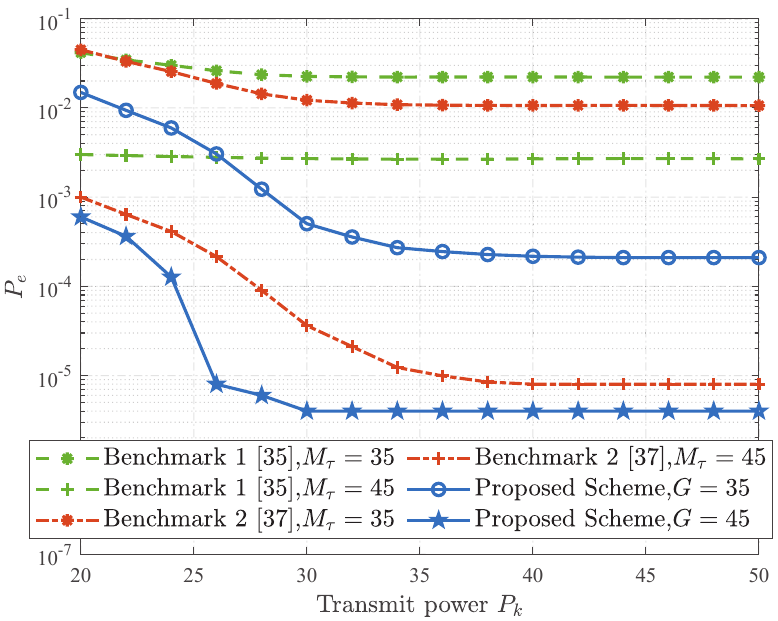}}	
	\subfigure[]{	
		\includegraphics[width=0.66\columnwidth, keepaspectratio]{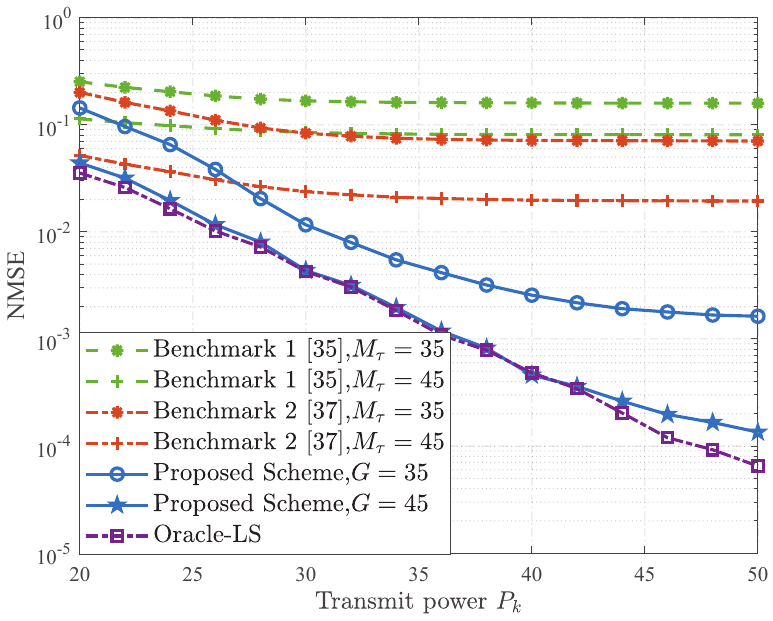}}	
	\subfigure[]{	
		\includegraphics[width=0.66\columnwidth, keepaspectratio]{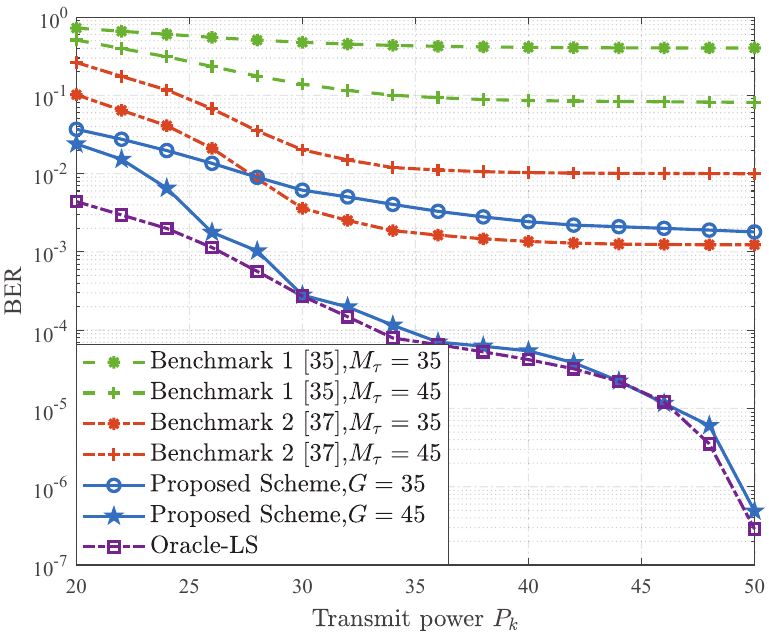}}	
	\caption{Performance comparison under different transmit power $P_k$: (a) $P_e$; (b) NMSE; (c) BER.}
	\label{fig:SNR}
\end{figure*}

\subsection{Performance under LoS TSL}  

First of all, we investigate the performance of ATI, CE, and SD under LoS TSL. 
As a fledgling concept, there has been little work dedicated to the field of ATI, CE, and multi-user SD in the framework of GF-NOMA-OTFS. 
{We take one of the most representative schemes proposed in \cite{SWQ.OTFS,GF-NOMA-OTFS} as the Benchmark 1 for comparison, which embeds the guard and non-orthogonal pilot symbols in the DD domain to facilitate uplink ATI and CE by exploring the sparsity of channel in the delay-Doppler-angle domain.}
The size of embedded DD domain pilots along the Doppler dimension and the delay dimension are denoted as $N_{\nu}$ and $M_{\tau}$, respectively, and $N$ is fixed as $N=N_{\nu}$. 
Besides, the size of embedded DD domain guard symbols along the Doppler dimension and the delay dimension, which are utilized to eliminate ISI, is set as $N_{g}=0$ and $M_{g}=L-1$, respectively [34, Fig.~34]. 
{The problem formulation and the adopted 3D-SOMP algorithm for comparison can be referred to \cite{SWQ.OTFS,GF-NOMA-OTFS}, respectively.}
Furthermore, we set Benchmark 2, where the virtual sampling grid in the DD space \cite{OTFS-fractional Doppler} is attached for the Benchmark 1 for further comparison, and its virtual sampling grid size in the Doppler domain is fixed as $N^{\prime} = 2N, M^{\prime} = M$.
In additional, the GF-NOMA scheme employing OFDM waveform without Doppler compensation, is considered for comparison as well \cite{KML}.  
Finally, the oracle-LS estimator with known ATS is used as performance upper bound.

\renewcommand{\arraystretch}{1.5}
\begin{table}[]
	\centering
	\caption{Transmission efficiency comparison between the benchmark and the proposed
		scheme}
	\label{tab:my-table4}
	\resizebox{\columnwidth}{!}{%
		\begin{tabular}{|c|cccc|cccc|}
			\hline
			\textbf{Schemes} & \multicolumn{4}{c|}{\textbf{Benchmark 1 and 2}} & \multicolumn{4}{c|}{\textbf{Proposed}} \\ \hline
			\multirow{2}{*}{\textbf{\begin{tabular}[c]{@{}c@{}}Cyclic prefix\\ (ISI region)\end{tabular}}} & \multicolumn{4}{c|}{$N(L-1)$} & \multicolumn{4}{c|}{$(N+1)(L-1)$} \\ \cline{2-9} 
			& \multicolumn{1}{c|}{256} & \multicolumn{1}{c|}{256} & \multicolumn{1}{c|}{256} & 256 & \multicolumn{1}{c|}{288} & \multicolumn{1}{c|}{288} & \multicolumn{1}{c|}{288} & 288 \\ \hline
			\multirow{2}{*}{\textbf{\begin{tabular}[c]{@{}c@{}}Guard \\ Interval\end{tabular}}} & \multicolumn{4}{c|}{$2N(L-1)$} & \multicolumn{4}{c|}{\multirow{2}{*}{\diagbox{\quad\quad\quad\quad}{}}} \\ \cline{2-5}
			& \multicolumn{1}{c|}{512} & \multicolumn{1}{c|}{512} & \multicolumn{1}{c|}{512} & 512 & \multicolumn{4}{c|}{} \\ \hline
			\multirow{2}{*}{\textbf{\begin{tabular}[c]{@{}c@{}}Effective \\ Pilot\end{tabular}}} & \multicolumn{4}{c|}{$M_{\tau}N_{\nu}$ ($N_{\nu}=N$)} & \multicolumn{4}{c|}{$G(N+1)$} \\ \cline{2-9} 
			& \multicolumn{1}{c|}{160} & \multicolumn{1}{c|}{240} & \multicolumn{1}{c|}{320} & 400 & \multicolumn{1}{c|}{180} & \multicolumn{1}{c|}{270} & \multicolumn{1}{c|}{360} & 450 \\ \hline
			\multirow{2}{*}{\textbf{Frame size}} & \multicolumn{4}{c|}{$(M+L-1)N$} & \multicolumn{4}{c|}{$M_t(N+1)+MN$} \\ \cline{2-9} 
			& \multicolumn{1}{c|}{2304} & \multicolumn{1}{c|}{2304} & \multicolumn{1}{c|}{2304} & 2304 & \multicolumn{1}{c|}{2516} & \multicolumn{1}{c|}{2606} & \multicolumn{1}{c|}{2696} & 2786 \\ \hline
			\multirow{2}{*}{\textbf{\begin{tabular}[c]{@{}c@{}}Transmission \\ efficeny\end{tabular}}} & \multicolumn{4}{c|}{$\frac{MN-2N(L-1)-M_{\tau}N_{\nu}}{(M+L-1)N}$} & \multicolumn{4}{c|}{$\frac{M(M+L-1)N^2}{[M_t(N+1)+MN]^2}$} \\ \cline{2-9} 
			& \multicolumn{1}{c|}{59.72\%} & \multicolumn{1}{c|}{56.25\%} & \multicolumn{1}{c|}{52.78\%} & 49.31\% & \multicolumn{1}{c|}{\textbf{74.54\%}} & \multicolumn{1}{c|}{\textbf{69.48\%}} & \multicolumn{1}{c|}{\textbf{64.92\%}} & \textbf{60.79\%} \\ \hline
		\end{tabular}%
	}
\end{table}

\begin{figure*}[tp]
	\centering
	\subfigure[]{	
		\includegraphics[width=0.66\columnwidth, keepaspectratio]{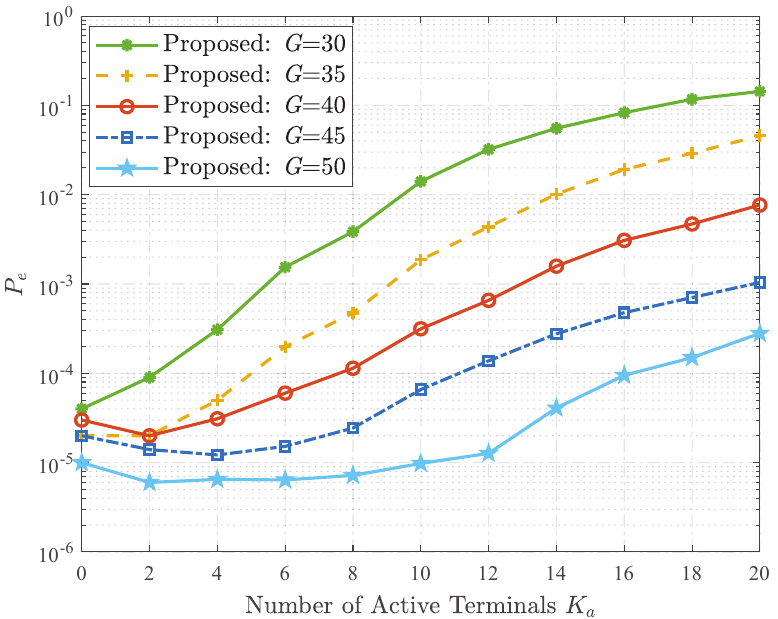}}
	\subfigure[]{	
		\includegraphics[width=0.66\columnwidth, keepaspectratio]{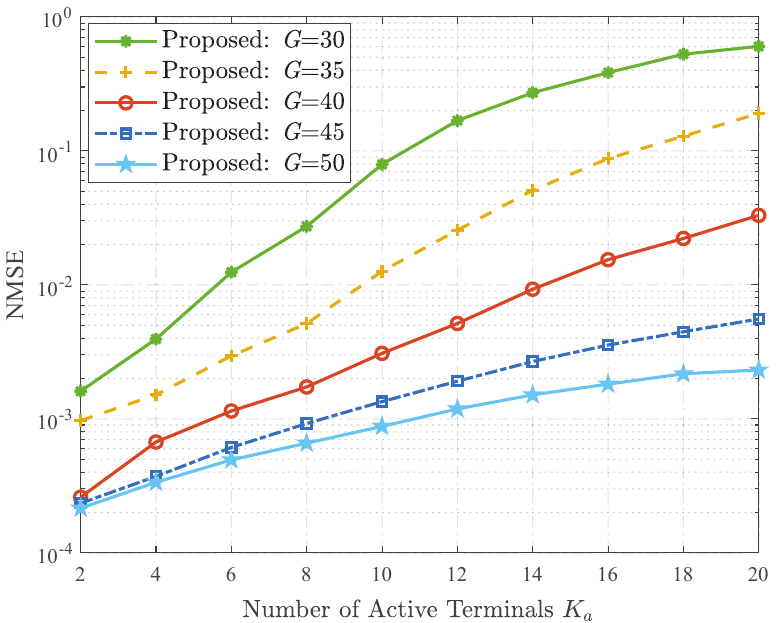}}
	\subfigure[]{	
		\includegraphics[width=0.66\columnwidth, keepaspectratio]{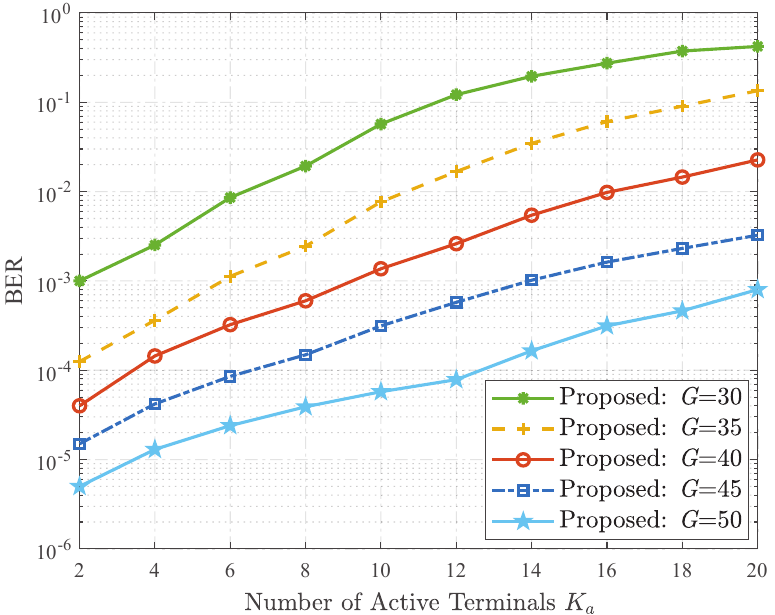}}	
	\caption{Performance under different number of active terminals $K_a$: (a) $P_e$; (b) NMSE; (c) BER.}
	\label{fig:K_a}
\end{figure*}
\begin{figure*}[tp]
	\centering	
	\subfigure[]{	
		\includegraphics[width=0.66\columnwidth, keepaspectratio]{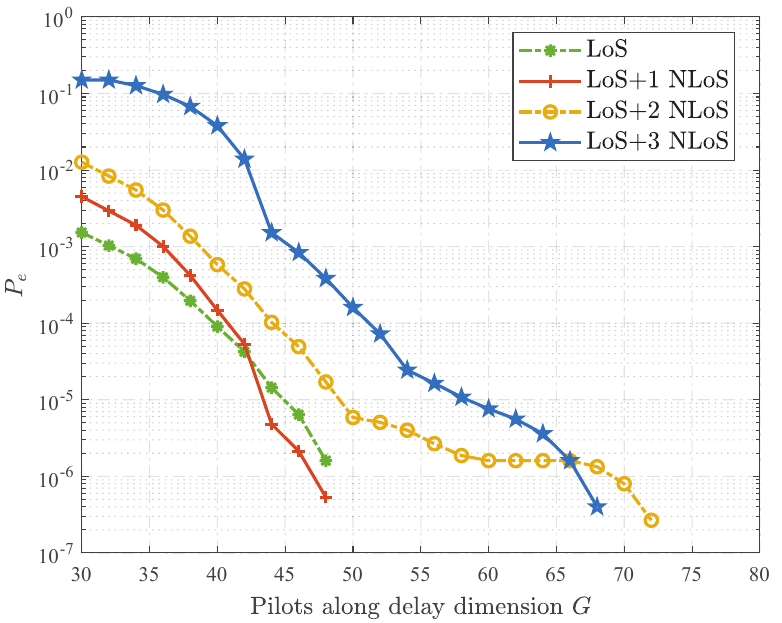}}	
	\subfigure[]{	
		\includegraphics[width=0.66\columnwidth, keepaspectratio]{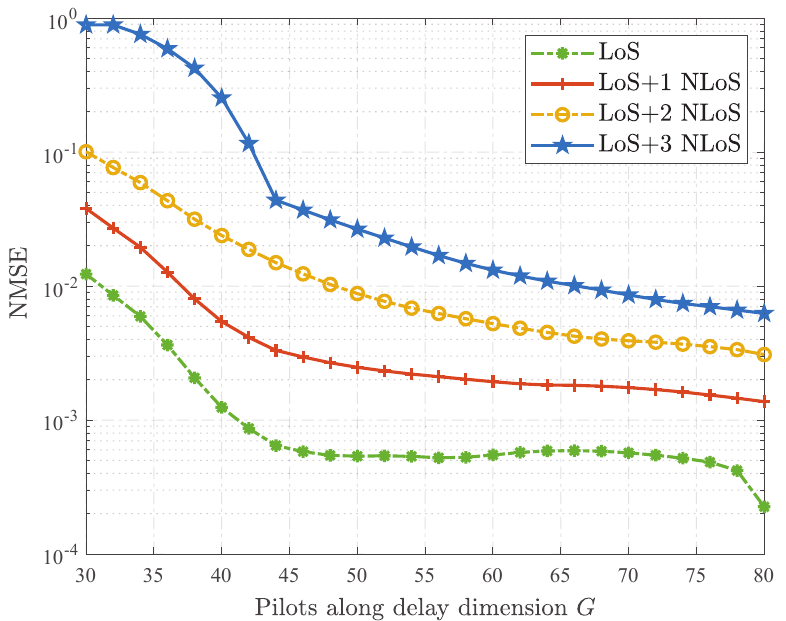}}	
	\subfigure[]{	
		\includegraphics[width=0.66\columnwidth, keepaspectratio]{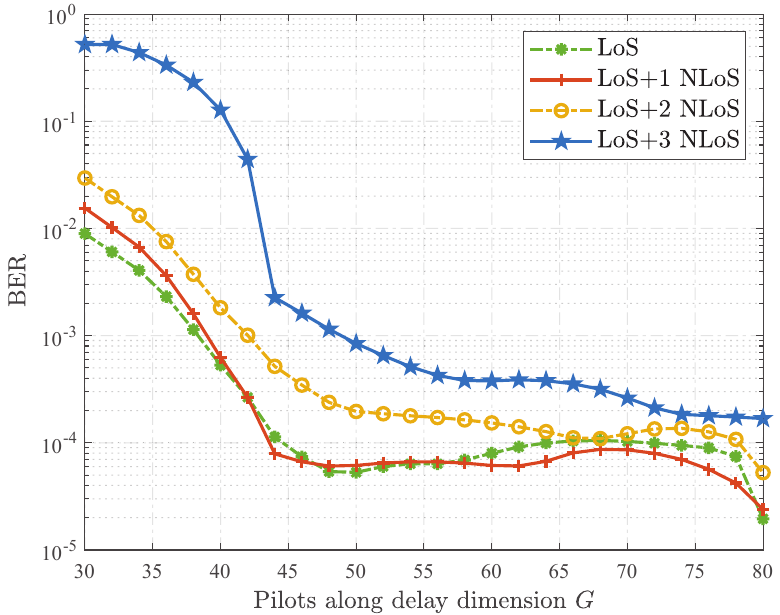}}	
	\caption{Performance under different channel conditions: (a) $P_e$; (b) NMSE; (c) BER.}
	\label{fig:NLoS}
\end{figure*}

Fig. \ref{fig:non-ISI size} provides $P_e$, NMSE, and BER performance under effective different pilot length,
where the effective pilot length is defined as follows: non-ISI region dimension $G$ for proposed scheme, DD domain pilots along delay dimension $M_{\tau}$ for Benchmark 1 and 2 \cite{SWQ.OTFS}, and time slots dimension $G$ occupied by pilots for GF-NOMA scheme employing OFDM waveform \cite{KML}. 
It can be observed from Fig. \ref{fig:non-ISI size} that, GF-NOMA scheme employing OFDM waveform suffers serious performance degradation confronted with severe Doppler effect when complicated compensation technique is absent. 
In contrast, the GF-NOMA-OTFS paradigms enjoy performance superiority owing to their Doppler robustness provided by OTFS.
It's noteworthy that the proposed scheme can further provide noticable performance gain over the Benchmark 1.   
To figure out the rationality behind this phenomenon, the result of Benchmark 2 is presented.
As Fig. \ref{fig:non-ISI size} exhibits, the superiority of Benchmark 2 over Benchmark 1 is self-evident, and it indicates that it's the low-resolution of Doppler domain that severely holds back the performance of Benchmark 1, especially when the small size of Doppler dimension $N$ is adopted.
However, oversize $N$ is prohibitive in the LEO satellite system for the intolerable computational complexity and signal processing latency. 
And more importantly, the quasi-static property of TSLs in the DD domain could be destroyed as $N$ increases. 
Therefore, the proposed scheme with the Doppler domain super-resolution enabled by the time domain TSs and parametric CE refinement is rewarding in this kind of harsh channel conditions.
Beisides, the NMSE and BER performance of the proposed method is very close to oracle-LS when the dimension of effective pilot overheads $G \ge 40$, which manifests that the approximation error of Eq. (\ref{eq:approx vec}) only leads to a slight increase of TSs overhead to ensure the performance of sparse signal recovery.
Meanwhile, the indisputable superiority of the proposed method even with lower TSs overhead demonstrates that the impact of approximation error on the following CE refinement is negligible in contrast to the low-resolution of Doppler domain of the Benchmark 1 and Benchmark 2.

Furthermore, to clearly present the percentage of the reduced pilot
overhead compared to Benchmark, we compare the transmission efficiency between the proposed scheme with Benchmark 1 and 2 as shown in Table \ref{tab:my-table4}, which is defined as the percentage of data symbols in the whole data frame. 
Therefore, we can conclude that the proposed scheme can reduce the pilot overhead  while achieving better performance.

Fig. \ref{fig:SNR} exhibits $P_e$, NMSE, and BER performance under different transmit power $P_k$.
It can be observed that the proposed scheme can achieve better $P_e$, NMSE, and BER performance while keeping the pilot overhead to a low level in the almost whole regime of transmit power (20-50 dBm) in contrast to Benchmark 1 and Benchmark 2. 
This can be interpreted that in the range of low transmit power, namely low SNR, the effective utilization of both the spatial and temporal correlations in the TSLs considerably promotes the accuracy of sparse signal recovery, and as a result, our proposed scheme outperforms the benchmarks. 
Moreover, in the range of high transmit power, namely high SNR, the BER performance is mainly dominated by the CE performance. 
In spite of the approximation error, our proposed scheme overcomes the problem of low-resolution in the Doppler domain and achieves a more satisfactory performance.

Moreover, since the traffic of mMTC is sporadic, the number of active terminals could be continuously varying. Besides, the number of active terminals $K_a$ is likely to get larger in massive MTC.  
In order to show the applicability of our proposed scheme in various IoT applications, we investigate $P_e$, NMSE, and BER performance under different number of active terminals $K_a$. The numerical results are illustrated in Fig. \ref{fig:K_a}.
On the one hand, when the number of active terminals $K_a=0$, $P_e$ degrades into the false alarming probability resulting from noise since there is no signal sent by the IoT terminals. 
On the other hand, when $K_a >0$, the performance of ATI, CE, and SD exhibit a similar deteriorating trend with an increasing number of active IoT terminals trying to access the LEO satellite in the same DD resources. 
Nevertheless, when appropriate TS overheads are employed, it could support a wide range of active terminals with tolerable performance losses. For instance, when $G = 50$, the performance of $P_e$, NMSE, and BER could still hold the superiority to $-30 \, \rm dB$, $-20 \, \rm dB$, and $-30 \, \rm dB$, respectively, when the number of active terminal varies from 2 to 20.

\subsection{Performance under Different Channel Conditions}

To further demonstrate the robustness of the proposed method, we investigate its performance under different TSL conditions.  
Fig. \ref{fig:NLoS} displays $P_e$, NMSE, and BER performance with the variation of MPCs, while the Rician factor is fixed at $\gamma_k=8 \, \mathrm{dB}, \forall k$. 
It can be observed that with the increase of MPCs, there is a slight rise of the effective pilot overheads to guarantee constant performance. 
This can be interpreted that the increase of MPCs leads to more observations to recover the increasing non-zero elements of sparse CIR vectors. 
In fact, despite the fact that the performance of NMSE deteriorates at a relatively rapid rate, the performance of $P_e$ and BER degrades sluggishly. 
It verifies the system performance is mainly determined by the accuracy of estimation of the LoS path and those low-energy NLoS paths have negligible impact on the system performance.
Besides, it is noteworthy that the increase of MPCs could contribute to the enhancement of ATI, which could be treated as a diversity gain.

\section{Conclusion}\label{S7}

This paper investigates an effective RA paradigm for accommodating massive IoT access based on LEO satellites. 
Specifically, we first propose to apply the GF-NOMA-OTFS scheme to LEO satellite-based IoT for mitigating the access scheduling overheads and latency, and combating the severe Doppler effect of TSLs. 
On this basis, to handle the challenging problem of ATI, CE, and SD, we further develop a TS-OTFS transmission scheme and a two-stage successive ATI and CE method. 
At the first stage, the time domain TSs facilitate us to leverage the traffic sparsity of IoT terminals and the sparse CIR to jointly perform ATI and coarse CE. 
Furthermore, a parametric approach is introduced to refine the CE performance based on the sparsity of TSLs in the DD domain.   
With the results of ATI and CE, we are further motivated to propose a time-domain parallel multi-user SD with relatively low computational complexity to circumvent the channel spreading in the DD or TF domain.
Simulation results demonstrate the effectiveness and superiority of our proposed paradigm particularly for LEO satellite-based massive access.

\begin{figure*}
	\begin{align}
		\hat{\mathbf{h}}^{{\rm eff-nz},i}_{{\rm TS}, p} 
		& = \boldsymbol{\Psi}_{[:, \mathcal{I}]}^{\dagger}   
		\sum_{k \in \mathcal{A}}  \sqrt{P_k}  \left(\mathbf{\Delta}_k^{\rm LoS} \boldsymbol{\psi}_{k}^{\rm LoS} h_{{\rm TS},k,p}^{{\rm eff}-i}(\ell_{k}^{\rm LoS}+1)  +
		\sum_{q=1}^{Q_k} \mathbf{\Delta}_k^{q}  \boldsymbol{\psi}_{k}^{q}  {h}_{{\rm TS},k,p}^{{\rm eff}-i}(\ell_k^q+1) \right)  + 
		\underbrace {\boldsymbol{\Psi}_{[:, \mathcal{I}]}^{\dagger} \mathbf{w}_{{\rm TS},p}^{i}}_{\mathrm{noise}}, \label{eq:long1} \tag{59}	\\	
		& \approx \boldsymbol{\Psi}_{[:, \mathcal{I}]}^{\dagger} \sum_{k \in \mathcal{A}}  \sqrt{P_k} \left[ \mathbf{\Delta}_k^{\rm LoS} \boldsymbol{\psi}_{k}^{\rm LoS}, \mathbf{\Delta}_k^{1}  \boldsymbol{\psi}_{k}^{1}, \dots , \mathbf{\Delta}_k^{Q_k}  \boldsymbol{\psi}_{k}^{Q_k} \right] \times 
		\left[ h_{{\rm TS},k,p}^{{\rm eff}-i}(\ell_{k}^{\rm LoS}+1) , {h}_{{\rm TS},k,p}^{{\rm eff}-i}(\ell_k^1+1)  , \dots , {h}_{{\rm TS},k,p}^{{\rm eff}-i}(\ell_k^{Q_k}+1) \right]^{\rm T}, \label{eq:long2} \tag{60}  \\			
		& \approx \boldsymbol{\Psi}_{[:, \mathcal{I}]}^{\dagger} \sum_{k \in \mathcal{A}} \sqrt{P_k}
		\left[ \mathbf{\Delta}_k^{\rm LoS} \boldsymbol{\psi}_{k}^{\rm LoS}, \mathbf{\Delta}_k^{1}  \boldsymbol{\psi}_{k}^{1}, \dots , \mathbf{\Delta}_k^{Q_k}  \boldsymbol{\psi}_{k}^{Q_k} \right] \times \nonumber \\  			
		& \quad\quad  \left[ {g_{k,p}^{\rm eff-LoS}}  e^{j2\pi\upsilon_k^{\rm LoS} \left[ \frac{i-1}{N}
			+ \frac{(L-\ell_k^{\rm LoS})}{N(M+M_t)} \right] }, 
		{g_{k,p}^{\rm eff-1}} e^{j2\pi\upsilon_k^{\rm 1} \left[ \frac{i-1}{N}
			+ \frac{(L-\ell_k^1)}{N(M+M_t)} \right] }, \dots,
		{g_{k,p}^{{\rm eff}-Q_k}}  e^{j2\pi\upsilon_k^{Q_k} \left[ \frac{i-1}{N}
			+ \frac{(L-\ell_k^{Q_k})}{N(M+M_t)} \right] }  
		\right]^{\rm T}, \label{eq:long3} \tag{61} \\	 
		& \approx \boldsymbol{\Psi}_{[:,\mathcal{I}]}^{\dagger} \sum_{k \in \mathcal{A}} \sqrt{P_k} 
		\underbrace{ \left[ \mathbf{\Delta}_k^{\rm LoS} \boldsymbol{\psi}_{k}^{\rm LoS}, \mathbf{\Delta}_k^{1}  \boldsymbol{\psi}_{k}^{1}, \dots , \mathbf{\Delta}_k^{Q_k}  \boldsymbol{\psi}_{k}^{Q_k} \right] }_{\mathbf{\Gamma}_k} \underbrace{\left[ g_{k,p}^{\rm eff-LoS},g_{k,p}^{\rm eff-1},\dots,g_{k,p}^{{\rm eff}-Q_k} \right]^{\rm T}}_{\mathbf{g}_{k,p}^{\rm eff}} \nonumber	\\ 
		& \quad\quad\quad\quad\quad\quad\quad\quad\quad  
		\odot \underbrace{ \left[
			e^{j2\pi\upsilon_k^{\rm LoS} \left[ \frac{i-1}{N}
				+ \frac{(L-\ell_k^{\rm LoS})}{N(M+M_t)} \right] }  , e^{j2\pi\upsilon_k^{\rm 1} \left[ \frac{i-1}{N}
				+ \frac{(L-\ell_k^1)}{N(M+M_t)} \right] } ,\dots,	e^{j2\pi\upsilon_k^{Q_k} \left[ \frac{i-1}{N}
				+ \frac{(L-\ell_k^{Q_k})}{N(M+M_t)} \right] }  \right]^{\rm T} }_{\boldsymbol{\eta}_k^{i-1}}.	\label{eq:long4} \tag{62}  
	\end{align}
	\hrulefill
\end{figure*}

\renewcommand{\appendixname}{Appendix\\Proof of Lemma 1}
\appendix

In fact, under the assumption that the support set of $\hat{\tilde{\mathbf{H}}}_{\rm TS}^{\rm eff}$ is perfectly recovered, the non-zero elements $\hat{\mathbf{h}}^{{\rm eff-nz},i}_{{\rm TS}, p}$ of  $\hat{\tilde{\mathbf{H}}}_{{\rm TS}_{[:,p+(i-1)P]}}^{\rm eff}$ can be derived from
\begin{align}\label{eq:SOMP}
	\hat{\mathbf{h}}^{{\rm eff-nz},i}_{{\rm TS}, p}  = \boldsymbol{\Psi}_{[:,\mathcal{I}]}^{\dagger}  \mathbf{r}_{{\rm TS},p}^{i}. \tag{58}	
\end{align}
From Eq. (\ref{eq:vec IO}), it can be further written as Eq. (\ref{eq:long1}), 
where $\boldsymbol{\psi}_k^{\rm LoS} = \boldsymbol{\Psi}_{k_{[:,\ell_k^{\rm LoS}+1]}}$ and $\boldsymbol{\psi}_k^q = \boldsymbol{\Psi}_{k_{[:,\ell_k^q+1]}}$. 
Ignoring the noise term, (\ref{eq:long1}) can be further approximate to the vector form as Eq. (\ref{eq:long2}).
According to the CIR model in Eq. (\ref{eq:T-CIR}) and  Eq. (\ref{eq:geff-L}), Eq. (\ref{eq:long2}) can be further expressed as Eq. (\ref{eq:long3}).
Finally, by extracting the effective channel coefficients,  Eq. (\ref{eq:long3}) can be decomposed into Eq. (\ref{eq:long4}). 
\setcounter{equation}{62}
Therefore, the mathematical relationship between $\hat{\mathbf{h}}^{{\rm eff-nz},i}_{{\rm TS}, p}$ and the effective channel coefficients can be represented by
\begin{align} \label{eq:48}
	\hat{\mathbf{h}}^{{\rm eff-nz},i}_{{\rm TS}, p}  = \boldsymbol{\Psi}_{[:,\mathcal{I}]}^{\dagger} \sum_{k \in \mathcal{A}}
	 \mathbf{\Gamma}_k \boldsymbol{\eta}_k^{i-1} \odot \mathbf{g}_{k,p}^{\rm eff} + \boldsymbol{\Psi}_{[:,\mathcal{I}]}^{\dagger} \mathbf{w}_{{\rm TS},p}^{i}. 
\end{align}
Furthermore,  by collecting the vectors and matrices with different subscripts $k$, Eq. (\ref{eq:48}) can be vectorized to 
\begin{align} \label{eq:fading factor}
	\hat{\mathbf{h}}^{{\rm eff-nz},i}_{{\rm TS}, p}  = \boldsymbol{\Psi}_{[:,\mathcal{I}]}^{\dagger}  \mathbf{\Gamma}  \boldsymbol{\eta}^{i-1} \odot \mathbf{g}_{p}^{\rm eff} + \boldsymbol{\Psi}_{[:,\mathcal{I}]}^{\dagger} \mathbf{w}_{{\rm TS},p}^{i}, 
\end{align}
where $\mathbf{\Gamma} = \left[  \mathbf{\Gamma}_{k_1},\mathbf{\Gamma}_{k_2},\dots,\mathbf{\Gamma}_{K_a} \right] \in \mathbb{C}^{G \times Q}$, $\boldsymbol{\eta}^{i-1} = \left[ \boldsymbol{\eta}^{i-1^{\rm T}}_{k_1},\boldsymbol{\eta}^{i-1^{\rm T}}_{k_2},\dots,\boldsymbol{\eta}^{i-1^{\rm T}}_{k_{K_a}} \right]^{\rm T} \in \mathbb{C}^{Q \times 1}$, and $\mathbf{g}^{\rm eff}_p = \left[ \mathbf{g}^{\rm eff^T}_{k_1,p},\mathbf{g}^{\rm eff^T}_{k_2,p},\dots,\mathbf{g}^{\rm eff^T}_{k_{K_a},p} \right]^{\rm T}  \in \mathbb{C}^{Q \times 1}$ with $k_1,k_1,\dots,k_{K_a} \in \mathcal{A}$.

It's clear that $\hat{\mathbf{h}}^{{\rm eff-nz},i}_{{\rm TS}, p}$ and $\mathbf{g}_{p}^{\rm eff}$ have linear relationship. 
Since $\mathbf{\Gamma}$ and $\boldsymbol{\eta}^{i-1}$ can be reconstructed with the estimated Doppler shift, RToA and MPCs' delay, $\mathbf{g}_{p}^{\rm eff}$ can be calculated mathematically based on the LS criterion according to Eq. (\ref{eq:fading factor}) as well. 
This completes the proof of {Lemma \ref{lemma1}}.

\begin{IEEEbiography}[{\includegraphics[width=1in,height=1.25in,clip,keepaspectratio]{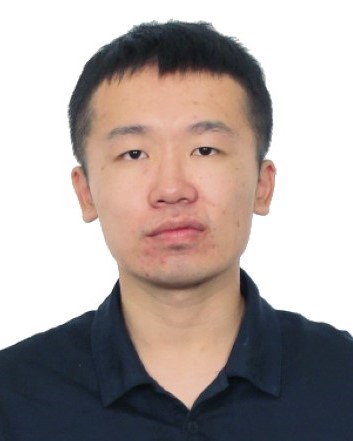}}]{Xingyu Zhou}
	received the B.S. degree from the School of Information and Electronics, Beijing Institute of Technology, Beijing, China, in 2021, where he is currently
	pursuing the M.S. degree.
	His research interests include space-air-ground-sea integrated networks, massive access, and OTFS waveform for the next generation wireless communications.		
\end{IEEEbiography}

\begin{IEEEbiography}[{\includegraphics[width=1in,height=1.25in,clip,keepaspectratio]{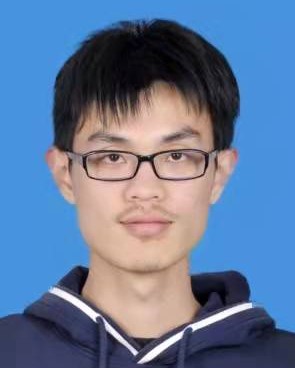}}]{Keke Ying}
	received the B.S. degree from the School of Information and Electronics, Beijing Institute of Technology, Beijing, China, in 2020, where he is currently
	pursuing the Ph.D. degree.
	His research interests include massive MIMO systems, satellite communications, and sparse signal processing.		
\end{IEEEbiography}

\begin{IEEEbiography}[{\includegraphics[width=1in,height=1.25in,clip,keepaspectratio]{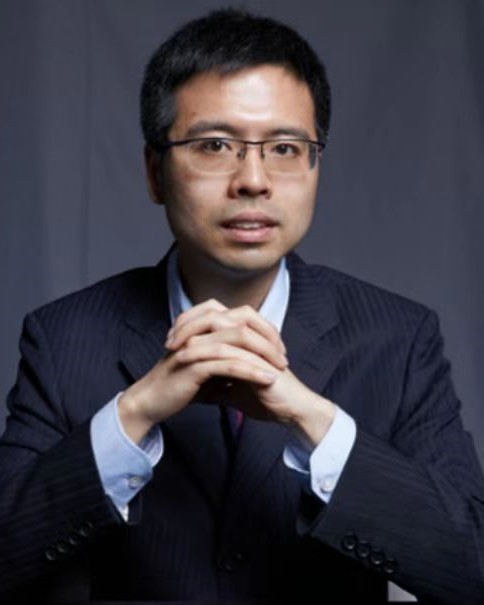}}]{Zhen Gao}
	received the B.S. degree in information engineering from the Beijing Institute of Technology, Beijing, China, in 2011, and the Ph.D.
	degree in communication and signal processing with the Tsinghua National Laboratory for Information Science and Technology, Department of Electronic
	Engineering, Tsinghua University, China, in 2016.
	He is currently an Assistant Professor with the Beijing Institute of Technology. His research interests are in wireless communications, with a focus on multi-carrier modulations, multiple antenna systems, and sparse signal processing.
	He was a recipient of the IEEE Broadcast Technology Society 2016 Scott Helt Memorial Award (Best Paper), the Exemplary Reviewer of {IEEE Communication Letters} in 2016, {IET Electronics Letters} Premium Award (Best Paper) 2016, and the Young Elite Scientists Sponsorship Program (2018--2021) from China Association for Science and Technology.
\end{IEEEbiography}

\begin{IEEEbiography}[{\includegraphics[width=1in,height=1.25in,clip,keepaspectratio]{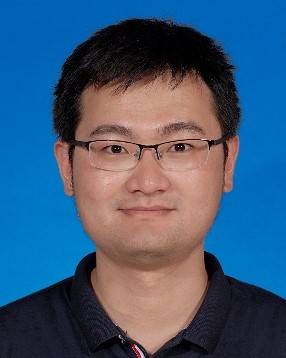}}]{Yongpeng Wu} (Senior Member, IEEE)
	received the B.S. degree in telecommunication engineering from
	Wuhan University, Wuhan, China, in July 2007, the Ph.D. degree in communication and signal processing from the National Mobile Communications Research Laboratory, Southeast University, Nanjing, China, in November 2013.

	He is currently a Tenure-Track Associate Professor
	with the Department of Electronic Engineering,
	Shanghai Jiao Tong University, Shanghai, China. Previously,
	he was Senior Research Fellow with the Institute for Communications Engineering, Technical University of Munich, Munich, Germany and the Humboldt Research Fellow and the Senior Research
	Fellow with the Institute for Digital Communications, University Erlangen-Nn{\"u}berg, Germany. During his doctoral studies, he conducted cooperative research
	with the Department of Electrical Engineering, Missouri University of Science
	and Technology, USA. His research interests include massive MIMO/MIMO
	systems, massive access, physical layer security, and power line communication.

	Dr. Wu was awarded the IEEE Student Travel Grants for IEEE International
	Conference on Communications 2010, the Alexander von Humboldt Fellowship
	in 2014, the Travel Grants for IEEE Communication Theory Workshop 2016, the
	Excellent Doctoral Thesis Awards of China Communications Society 2016, the
	Exemplary Editor Award of IEEE Communications Letters 2017, and Young
	Elite Scientist Sponsorship Program by CAST 2017. He was an Exemplary
	Reviewer of the IEEE Transactions on Communications in 2015, 2016,
	and 2018, respectively. He was the lead Guest Editor for the Special Issue
	Physical Layer Security for 5G Wireless Networks of IEEE Journal on
	Selected Areas in Communications and the Guest Editor for the Special
	Issue Safeguarding 5G-and-Beyond Networks with Physical Layer Security
	of IEEE Wireless Communications. He is currently an Editor for the IEEE
	Transactions on Communications and IEEE Communications Letters.
	He has been a TPC member of various conferences, including Globecom, ICC,
	VTC, and PIMRC, etc.
\end{IEEEbiography}

\begin{IEEEbiography}[{\includegraphics[width=1in,height=1.25in,clip,keepaspectratio]{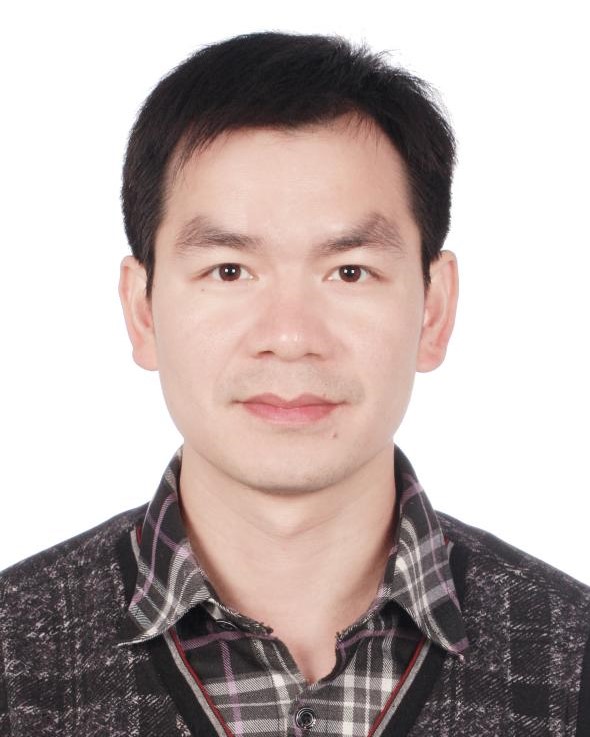}}]{Zhenyu Xiao} (Senior Member, IEEE)
received the B.E. degree from the Department of Electronics and Information Engineering, Huazhong University of Science and Technology, Wuhan, China, in 2006, and the Ph.D. degree from the Department of Electronic Engineering, Tsinghua University, Beijing, China, in 2011.

From 2011 to 2013, he held a postdoctoral position with the Department of Electronic Engineering, Tsinghua University. He was with the School of Electronic and Information Engineering, Beihang University, Beijing, as a Lecturer from 2013 to 2016, and an Associate Professor from 2016 to 2020, where he is currently a Full Professor.
He has visited the University of Delaware from 2012 to 2013, and the Imperial College London from 2015 to 2016. He has authored or coauthored over 70 papers, including IEEE Journal on Selected Areas in Communications, IEEE Transactions on Wireless
Communications, IEEE Transactions on Signal Processing, IEEE
Transactions on Vehicular Technology, IEEE Communications Letters, IEEE Wireless Communications Letters, and IET Communications. He has received the 2017 Best Reviewer Award of IEEE Transactions on Wireless Communications, the 2019 Exemplary
Reviewer Award of IEEE Wireless Communications Letters, and the 4th China Publishing Government Award. He has received the Second Prize of National Technological Invention, the First Prize of Technical Invention of China Society of Aeronautics and Astronautics, and the Second Prize of Natural Science of China Electronics Society. He is an active Researcher with broad interests on millimeter wave communications and UAV/satellite
communications and networking. He was elected as one of the 2020 Highly
Cited Chinese Researchers. He is currently an Associate Editor for IEEE Transactions on Cognitive Communications and Networking, China Communications, IET Communications, KSII Transactions on Internet and Information Systems, and Frontiers in Communications and Networks.
He has also been a Lead Guest Editor of a special issue named ``Antenna
Array Enabled Space/Air/Ground Communications and Networking'' of IEEE Journal on Selected Areas in Communications, one named ``Space-Air-Ground Integrated Network with Native Intelligence (NI-SAGIN): Concept, Architecture, Technology, and Radio'' of China Communications, and one named ``LEO Satellite Constellation Networks'' of Frontiers in
Communications and Networks. He has been a TPC Member of IEEE GLOBECOM, IEEE WCSP, IEEE ICC, and IEEE ICCC.	
\end{IEEEbiography}

\begin{IEEEbiography}[{\includegraphics[width=1in,height=1.25in,clip,keepaspectratio]{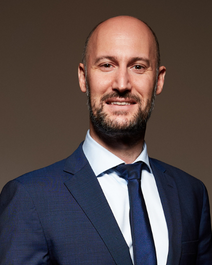}}]{Symeon Chatzinotas}
	(Senior Member, IEEE)
	received the M.Eng. degree in telecommunications
	from the Aristotle University of Thessaloniki, Thessaloniki, Greece, in 2003, and the M.Sc. and Ph.D. degrees in electronic engineering from the
	University of Surrey, Surrey, U.K., in 2006 and
	2009, respectively. He is currently a Full Professor/a
	Chief Scientist and the Co-Head with the SIGCOM
	Research Group, SnT, University of Luxembourg.
	In the past, he has been a Visiting Professor at
	the University of Parma, Italy. He is coordinating
	the research activities on communications and networking, acting as a PI
	for more than 20 projects and main representative for 3GPP, ETSI, and
	DVB. He was involved in numerous research and development projects for
	NCSR Demokritos, CERTH Hellas, CCSR, and the University of Surrey.
	He has (co)authored more than 500 technical papers in refereed international
	journals, conferences, and scientific books. He was the co-recipient of the 2014
	IEEE Distinguished Contributions to Satellite Communications Award and the
	Best Paper Awards at EURASIP JWCN, CROWNCOM, and ICSSC. He is
	currently serving in the editorial board for the IEEE Transactions on Communications, IEEE Open Journal of Vehicular Technology, and the International Journal of Satellite Communications and Networking.
\end{IEEEbiography}

\begin{IEEEbiography}[{\includegraphics[width=1in,height=1.25in,clip,keepaspectratio]{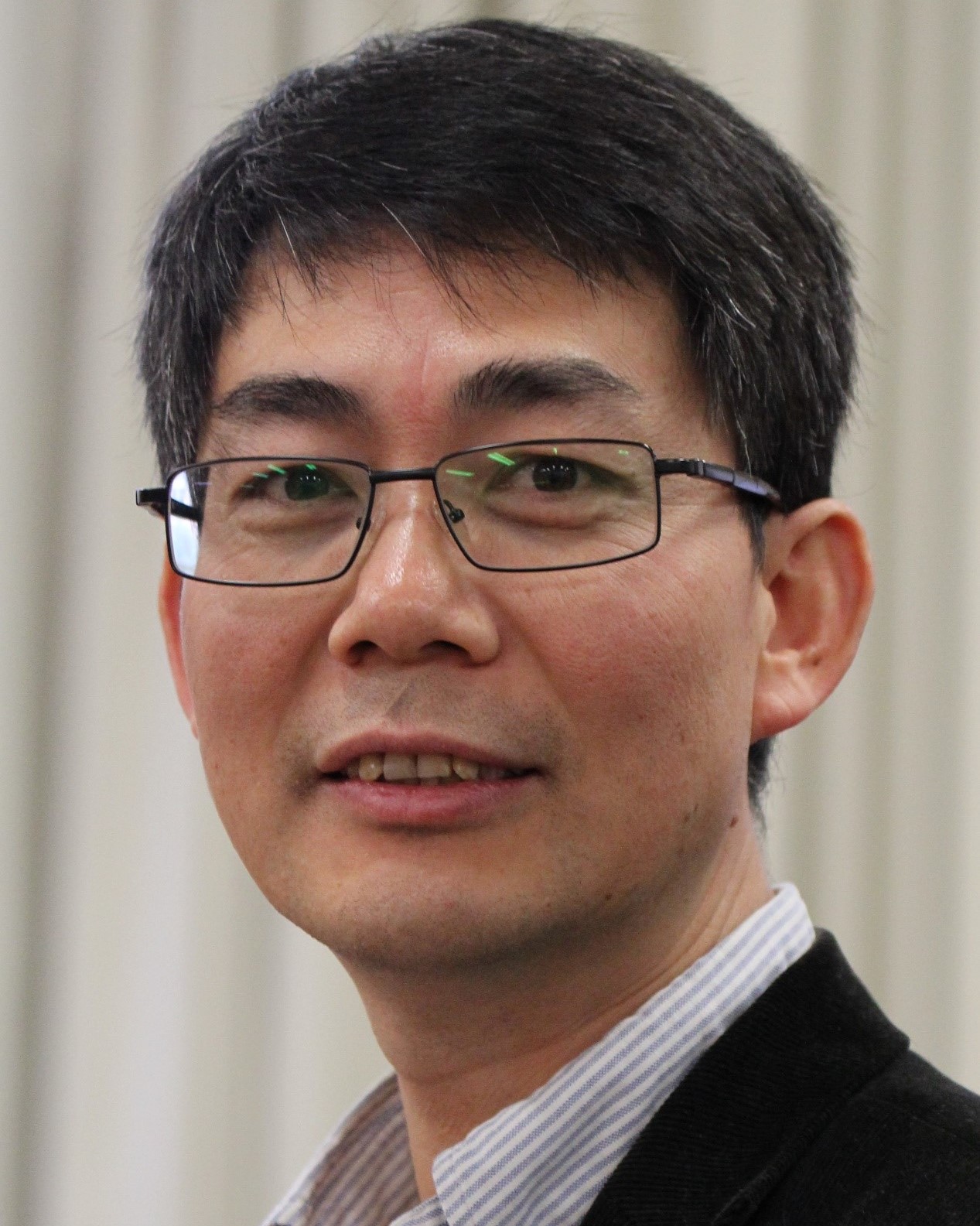}}]{Jinhong Yuan}
	(Fellow, IEEE) received the B.E. and
	Ph.D. degrees in electronics engineering from the
	Beijing Institute of Technology, Beijing, China, in
	1991 and 1997, respectively. From 1997 to 1999,
	he was a Research Fellow with the School of Electrical
	Engineering, The University of Sydney, Sydney,
	Australia. In 2000, he joined the School of Electrical
	Engineering and Telecommunications, University of
	New South Wales, Sydney, where he is currently
	a Professor and the Head of the Telecommunication
	Group, School of Electrical Engineering and
	Telecommunications. He has published two books, five book chapters, over
	300 articles in telecommunications journals and conference proceedings, and
	50 industrial reports. He is a co-inventor of one patent on MIMO systems and
	four patents on low-density-parity-check codes. His current research interests
	include error control coding and information theory, communication theory,
	and wireless communications. He has coauthored four best paper awards
	and one best poster award, including the Best Paper Award from the IEEE
	International Conference on Communications, Kansas City, USA, in 2018;
	the Best Paper Award from IEEE Wireless Communications and Networking
	Conference, Cancun, Mexico, in 2011; and the Best Paper Award from
	the IEEE International Symposium on Wireless Communications Systems,
	Trondheim, Norway, in 2007. He has served as the IEEE NSW Chapter
	Chair for the Joint Communications/Signal Processions/Ocean Engineering
	Chapter from 2011 to 2014. He has served as an Associate Editor for the
	IEEE Transactions on Communications from 2012 to 2017. He is
	serving as an Associate Editor for the IEEE Transactions on Wireless
	Communications and IEEE Transactions on Communications.
\end{IEEEbiography}

\begin{IEEEbiography}[{\includegraphics[width=1in,height=1.25in,clip,keepaspectratio]{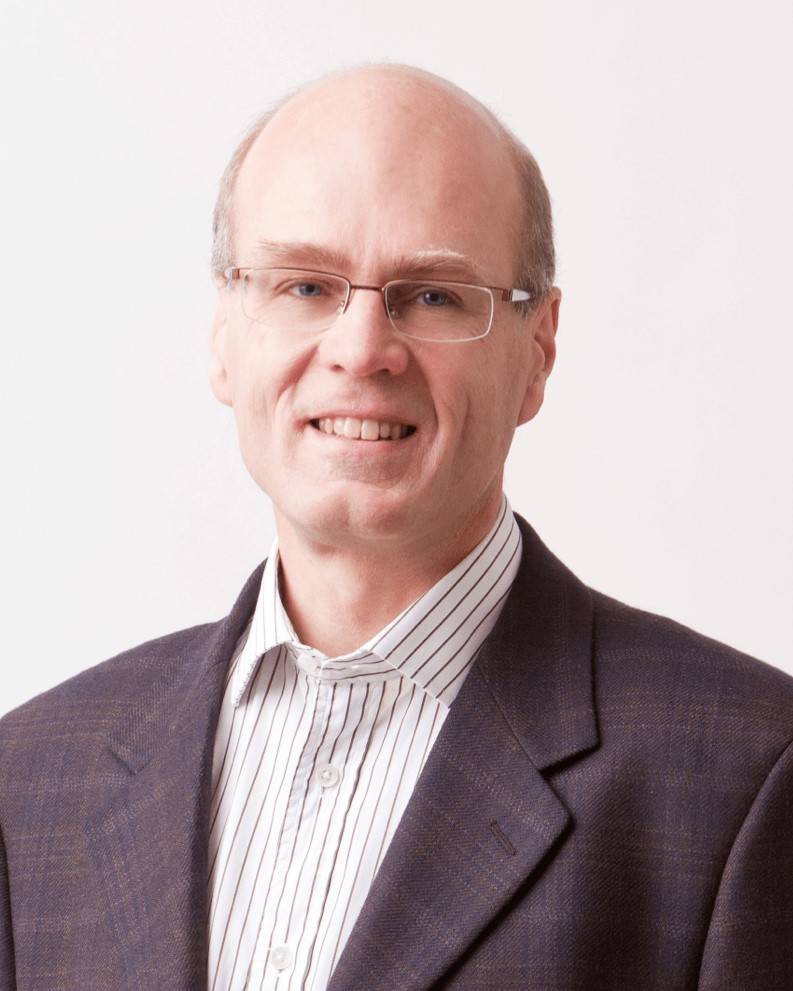}}]{Bj\"{o}rn Ottersten}
	(Fellow, IEEE) received the M.S.
	degree in electrical engineering and applied physics
	from Link\"{o}ping University, Link\"{o}ping, Sweden,
	in 1986, and the Ph.D. degree in electrical engineering
	from Stanford University, Stanford, CA,
	USA, in 1990. He has held research positions
	with the Department of Electrical Engineering,
	Link\"{o}ping University; the Information Systems Laboratory,
	Stanford University; the Katholieke Universiteit
	Leuven, Leuven, Belgium; and the University
	of Luxembourg, Luxembourg. From 1996 to 1997,
	he was the Director of Research with ArrayComm, Inc., a start-up in San Jose,
	CA, USA, based on his patented technology. In 1991, he was appointed as a
	Professor of signal processing with the Royal Institute of Technology (KTH),
	Stockholm, Sweden. He has been the Head with the Department for Signals,
	Sensors, and Systems, KTH, and the Dean with the School of Electrical
	Engineering, KTH. He is currently the Director for the Interdisciplinary Centre
	for Security, Reliability and Trust, University of Luxembourg. He is a fellow of
	EURASIP. He was a recipient of the IEEE Signal Processing Society Technical
	Achievement Award, the EURASIP Group Technical Achievement Award, and
	the European Research Council advanced research grant twice. He has coauthored
	journal papers that received the IEEE Signal Processing Society Best
	Paper Award in 1993, 2001, 2006, 2013, and 2019, and eight IEEE conference
	papers best paper awards. He has been a Board Member of IEEE Signal
	Processing Society and the Swedish Research Council. He currently serves for
	the boards of EURASIP and the Swedish Foundation for Strategic Research.
	He has served as the Editor-in-Chief of EURASIP Journal on Advances in
	Signal Processing and acted on the editorial boards of IEEE Transactions
	on Signal Processing, IEEE Signal Processing Magazine, IEEE Open
	Journal of Signal Processing, EURASIP Journal on Advances in Signal
	Processing, and Foundations and Trends in Signal Processing.
\end{IEEEbiography}

\end{document}